\documentclass[authoryear,review,11pt,a4paper]{elsarticle}

\usepackage{setspace}
\usepackage{fullpage}

\usepackage{color}

\usepackage{amsmath}

\DeclareMathOperator*{\argmin}{arg\,min}
\usepackage{amsfonts}
\usepackage{amssymb}
\usepackage{amsthm}

\usepackage{amscd}

\usepackage{bbm}

\usepackage{algorithm}
\usepackage{algpseudocode}

\usepackage{graphicx}

\usepackage{lscape}

\usepackage{booktabs}
\usepackage{multirow}
\usepackage{float}
\usepackage{caption}
\usepackage{subcaption}

\usepackage{lineno}

\usepackage{hyperref}

\usepackage{listings}
\lstnewenvironment{rc}[1][]{\lstset{language=R}}{}
\biboptions{authoryear}

\newtheorem{prop}{Proposition}
\newproof{pf}{Proof}

\journal{ }

\begin{document}

\begin{frontmatter}

\title{Joint optimisation of interim timing and stopping boundaries in group sequential trials}

\author[TGIUK,ICTU]{Zhangyi He\corref{cor1}\fnref{fn3}}
\ead{zhe1@georgeinstitute.org.uk}

\author[BrisMath]{Feng Yu}

\author[ICTU]{Suzie Cro}

\author[TGIAUS]{Laurent Billot}

\cortext[cor1]{Corresponding author.}
\fntext[fn3]{Affiliation at time of study}

\address[TGIUK]{The George Institute for Global Health, Imperial College London, London W12 7RZ, United Kingdom}
\address[ICTU]{Imperial Clinical Trials Unit, Imperial College London, London W12 7RH, United Kingdom}
\address[BrisMath]{School of Mathematics, University of Bristol, Bristol BS8 1UG, United Kingdom}
\address[TGIAUS]{The George Institute for Global Health, University of New South Wales, Sydney, NSW 2031, Australia}

\begin{abstract}
Group sequential designs (GSDs) are common in confirmatory trials, and their efficiency depends jointly on the timing of interim analyses and the shape of stopping boundaries. Most existing approaches optimise one component while holding the other fixed. Joint optimisation over free critical values has been studied for two- and three-stage efficacy-only designs and for tests with binding futility boundaries. Here we jointly optimise the full interim schedule and the boundary shape within the Hwang--Shih--DeCani and Kim--DeMets spending families, minimising the expected sample size under a specified alternative while controlling overall type I and type II error rates. A convex error-spending constraint prevents excessive early type I error spending. On the canonical information scale, the joint optimum depends only on the error rates, number of analyses and spending family, and is invariant to endpoint type and effect size. Without the convexity constraint, the parametric optima reproduced the unrestricted optima of \citet{brittain1993} to within $0.001$ in expected inflation factor. For three or more looks the jointly optimal design Pareto-dominated the scheduling-only optimal Pocock design, achieving the same or smaller expected sample size with uniformly smaller maximum sample size. Under a prespecified 1\% tolerance, six analyses, five interim plus a final, captured essentially all efficiency attainable with up to seven analyses, reducing expected sample size by 23\% and 33\% at 80\% and 90\% power relative to the fixed-sample design. We illustrate the approach by redesigning the CORRECT trial.
\end{abstract}


\end{frontmatter}


\section{Introduction}
\label{sec:1}
Group sequential design (GSD) is the most widely used adaptive design in confirmatory trials \citep{hatfield2016}. It allows a trial to stop early for efficacy or futility at preplanned interim analyses without inflating the overall type I error rate \citep{pocock1977,obrien1979,lan1983,kim1987,hwang1990}. Its efficiency depends on two interrelated design choices that jointly determine its operating characteristics. The first is how to shape the stopping boundaries, which is usually characterised by the error-spending functions that allocate type I and type II errors across analyses. The second is when to conduct the interim looks, which is commonly characterised by the information fractions at which these analyses are performed. Together, these choices determine the probability of early stopping at each interim analysis, the maximum sample size (MSS) required, and the expected sample size (ESS) under a given treatment effect size \citep{meurer2021}. Throughout, we also express these quantities on a relative scale using the maximum inflation factor (MIF) and the expected inflation factor (EIF), defined as the ratios of the MSS and ESS, respectively, to the sample size required by the corresponding fixed-sample design. 

Optimal GSDs have been extensively studied, primarily through optimisation of the stopping boundaries for prespecified timings of interim analyses \citep[\textit{e.g.},][]{wang1987,eales1992,chang1996,barber2002,anderson2007,wason2012b,ham2026}. \citet{grayling2021} examined how the optimal timing of interim analyses changes when the variability of the required sample size is taken into account. A complementary line of work has instead fixed the stopping boundaries and optimised the timing of interim analyses. \citet{togo2013} investigated how to time a single interim analysis to minimise the ESS under a specified alternative hypothesis for the Pocock design \citep{pocock1977} and the O'Brien--Fleming (OBF) design \citep{obrien1979}, each implemented via the Lan--DeMets spending function \citep{lan1983}. \citet{xi2017} chose the timing of futility interim analyses jointly with the futility boundary, but their focus was restricted to futility stopping. More recently, \citet{he2025} introduced \texttt{OptimInterim}, an \texttt{R}-based tool that optimises the scheduling of interim looks for given stopping boundaries, minimising the ESS under a specified alternative hypothesis subject to nominal type I and type II error rates. Neither line, however, optimises the full interim schedule jointly with the shape of the efficacy boundary, suggesting that further reductions in ESS may be achievable.

Joint optimisation of the timing of analyses and the stopping boundaries has also been studied. \citet{wason2012a} chose the stage sample sizes of a two-stage design together with the stopping thresholds. \citet{eales1992} derived optimal one-sided tests that stop to accept or reject the null hypothesis, solving a Bayes sequential decision problem by dynamic programming for given group sizes and searching over the group sizes for up to five groups. With equal type I and type II error rates, freeing the group sizes reduced the ESS at the alternative only slightly once the critical values were optimal. \citet{brittain1993} optimised testing times and critical values together, tabulating two- and three-stage efficacy-only designs that minimise the ESS under a specified alternative hypothesis subject to given type I and type II error rates. \citet{jennison2006} compared tests with binding futility boundaries for up to six analyses, using the ESS averaged over the null hypothesis, the alternative and a more optimistic effect size. The average ESS of tests with a monomial error spending function and equal group sizes except the first exceeded that of tests with freely optimised group sizes and critical values by at most 2\% of the fixed sample size. This indicates that restricting the boundary to a low-dimensional family need not substantially reduce efficiency. For boundary optimisation alone, \citet{anderson2007} restricted the boundaries to one-parameter error spending function families \citep{kim1987,hwang1990}, and \citet{wason2015} to a two-parameter extension of the power family of stopping boundaries \citep{pampallona1994}. Each governs the boundary shape with a small number of parameters. Such families retain substantial flexibility in boundary shape while greatly reducing the dimension of the search space, and hence the computational cost of optimisation. To our knowledge, the full interim schedule has not previously been optimised jointly with the boundary shape within such families.

A further challenge in the joint optimisation setting is determining which shapes of stopping boundaries are clinically credible. Because early interim analyses rest on limited accumulated information, allocating a large proportion of type I error to them increases the risk of false-positive stopping decisions driven by unstable estimates. This motivates a convex alpha-spending constraint, under which the rate at which type I error is spent is non-decreasing in information time.

Recently, \citet{li2023} illustrated by simulation that limiting a GSD to one or two interim looks can be inefficient, as it reduces the opportunity for early stopping. Allowing more interim analyses can therefore improve efficiency, although the gains diminish as the number of analyses increases. This is consistent with the diminishing returns reported by \citet{he2025}. From an administrative perspective, planning for between four and eight interim analyses is generally considered sensible \citep{todd2001}. Moreover, the number of analyses required to capture most of the available efficiency gain is seldom clear in advance. Over-specification has drawbacks, because each additional analysis typically requires a larger MSS and a further round of data monitoring, and beyond a certain point yields only a negligible reduction in the ESS. Over-specification may also, in principle, produce adjacent looks so close together that they become operationally indistinguishable. A procedure for selecting the effective number of looks should therefore be able to identify and merge such redundant stages.

To address these challenges, we jointly optimise the scheduling of interim analyses and the shape of stopping boundaries in GSDs, aiming to minimise the ESS under a specified alternative hypothesis while maintaining nominal type I and type II error rates. The stopping boundaries are parameterised through one-parameter spending function families, namely the Hwang--Shih--DeCani (HSD, Gamma) family of \citet{hwang1990} and the Kim--DeMets (KD, Power) family of \citet{kim1987}, with a second shape parameter when a futility boundary is included. A convex error-spending constraint is imposed to exclude clinically implausible designs with excessive early error spending. To avoid over-specifying the number of analyses, we further introduce a stage-merging rule and an automatic criterion for selecting the effective number of analyses. If two successive analyses of an optimised design lie closer together than a prespecified fraction of the total information, the stage-merging rule replaces that design with the optimum for fewer looks, reducing the number until no such gap remains. The rule is applied after optimisation rather than imposed as a constraint within it. The selection criterion then chooses the smallest number of analyses whose ESS lies within a prespecified tolerance of the smallest ESS attainable. We evaluate the operating characteristics of jointly optimised GSDs with up to six interim looks and benchmark the parametric joint optimum, with and without the convexity constraint, against the unrestricted optima of \citet{brittain1993}. We then illustrate the practical application of the proposed approach through a redesign of the CORRECT trial \citep{grothey2013} in previously treated metastatic colorectal cancer.


The scope of this paper is the statistical optimisation of a prespecified analysis grid together with the shape of the efficacy and futility boundaries, for a design fixed before the trial begins. It does not address data-dependent modification of the realised schedule once the trial is under way, nor inference following such modifications. It also leaves aside the geometry of futility rules defined by a common threshold on the observed effect, and the selection of a design under a sponsor-imposed budget on the maximum sample size. Those are distinct problems and are not considered here.

\section{Materials and Methods}
\label{sec:2}
We first review the GSD framework and the scheduling-only optimisation of \citet{he2025}, and then formalise the joint optimisation of interim scheduling and stopping boundaries under a convex error-spending constraint. We describe the stage-merging rule combined with automatic selection of the effective stage number, present the extended \texttt{OptimInterim}, and outline the numerical studies and the CORRECT case study used to evaluate the approach. 

\subsection{Group sequential designs}
\label{sec:21}
We adopt the framework and notation of \citet{he2025}. Consider a group sequential trial with $K>1$ analyses comparing a treatment with a control, where $\delta$ denotes the treatment effect and superiority is established by rejecting $H_{0}:\delta \leq 0$ in favour of $H_{1}:\delta>0$. The overall one-sided type I error rate is $\alpha$ and the type II error rate is $\beta$. The primary outcome $Y$ is assumed to be normally distributed with known common variance $\sigma^{2}$. We further assume that recruitment is suspended at each interim analysis, with all outcome data observed immediately. 
Throughout, $K$ denotes the total number of analyses, comprising $K-1$ interim analyses and the final analysis. 
At the $k$-th analysis, let $N_{t,k}$ and $N_{c,k}$ denote the cumulative numbers of participants with observed outcomes in the treatment and control arms. The information is then $I_{k}=N_{t,k}N_{c,k}/[\sigma^{2}(N_{t,k}+N_{c,k})]$, and the information rate is $t_{k}=I_{k}/I_{K}$, with $0<t_{1}<t_{2}<\cdots<t_{K}=1$. The standardised test statistic $Z_{k}$ follows the canonical multivariate normal distribution with mean $\mu_{k}=\delta\sqrt{I_{k}}$ and covariance $\Sigma_{kk'}=\sqrt{\min(I_{k},I_{k'})/\max(I_{k},I_{k'})}$. We write $\boldsymbol{\mu}_{1:k}$ and $\boldsymbol{\Sigma}_{k \times k}$ for the restricted mean vector and covariance matrix.

At each interim analysis ($k<K$), the trial stops early for efficacy if $Z_{k} \geq u_{k}$ or for futility if $Z_{k} \leq l_{k}$ (binding or non-binding), where $u_{k}$ and $l_{k}$ represent the efficacy and futility stopping boundaries at the $k$-th analysis. Following \citet{he2025}, we define the continuation, rejection and futility regions at the $k$-th analysis as $\mathcal{C}_{k}=\{z \mid l_{k}<z<u_{k}\}$, $\mathcal{R}_{k}=\{z \mid z \geq u_{k}\}$, and $\mathcal{F}_{k}=\{z \mid z \leq l_{k}\}$ with $l_{k}=-\infty$ when no futility boundary is specified, and $l_{K}=u_{K}$ so that $\mathcal{C}_{K}=\emptyset$. These three regions partition the real line at every analysis, and there is a single continuation region $\mathcal{C}_{k}$ governing the conduct of the trial, irrespective of which hypothesis is true. The region entering the calibration of the type I error rate need not be $\mathcal{C}_{k}$, however. Under a binding futility rule the sponsor commits in advance to stopping whenever $Z_{k} \in \mathcal{F}_{k}$, so such sample paths cannot contribute to the type I error rate, whereas under a non-binding rule they may. We therefore write
\begin{linenomath}
	\begin{equation*}
		\mathcal{C}^{\alpha}_{k}
		=
		\begin{cases}
			\mathcal{C}_{k}, & \text{for binding futility,} \\
			\{z \mid z<u_{k}\}, & \text{for non-binding or absent futility,}
		\end{cases}
	\end{equation*}
\end{linenomath}
for the region used in the type I error calculation, while the type II error rate and the ESS are computed over the operating region $\mathcal{C}_{k}$ itself. The superscript indexes the error rate being computed, not the hypothesis assumed to hold. The type I and type II error allocations at the $k$-th look are $\alpha_{k}$ and $\beta_{k}$. 
The stopping boundaries are calibrated through spending functions to preserve the nominal overall type I error rate \citep{lan1983,kim1987,hwang1990}. Details of the calibration procedure are given in \citet{he2025}.

Although the development below is written for a continuous normal outcome, the canonical joint distribution of $\boldsymbol{Z}_{1:K}=(Z_{1},Z_{2},\ldots,Z_{K})^\intercal$ also holds asymptotically for other endpoint types under standard regularity conditions. For binary endpoints, the Wald or score statistic for the difference in proportions is asymptotically normal, with the Fisher information at the $k$-th analysis determined by the cumulative sample size and the control event rate through the standard binary sample-size formula \citep{wassmer2024}. For count endpoints, the Wald statistic for the log rate ratio under a negative binomial model is asymptotically normal, with the Fisher information at the $k$-th analysis determined by the cumulative sample size, exposure times, event rates and overdispersion parameter \citep{mutze2019}. For survival endpoints under proportional hazards, the log-rank statistic is asymptotically normal, with $I_{k}=d_{k}r/(1+r)^{2}$ for allocation ratio $r$ and cumulative events $d_{k}$ \citep{schoenfeld1981,schoenfeld1983}, so that $t_{k}=d_{k}/d_{K}$ is driven by event accumulation rather than participant enrolment. 

\subsection{Joint optimisation of interim scheduling and stopping boundaries}
\label{sec:22}
We let $z_{p}$ denote the $p$-th quantile of the standard normal distribution, $\delta^{*}$ the target effect size under $H_{1}$ and $r$ the allocation ratio between the treatment and control arms. The ratio $r$ is held fixed across analyses, so that $I_{k}$ is proportional to the cumulative sample size. The sample size required for the corresponding fixed-sample design is
\begin{linenomath}
	\begin{equation*}
		N_{0}
		=
		\biggl(\frac{\sigma}{\delta^{*}}\biggr)^{2} \frac{(1+r)^{2}}{r}(z_{1-\alpha}+z_{1-\beta})^{2}.
	\end{equation*}
\end{linenomath}
We denote by $\boldsymbol{\alpha}_{1:K}=(\alpha_{1},\alpha_{2},\ldots,\alpha_{K})^{\intercal}$ and $\boldsymbol{\beta}_{1:K}=(\beta_{1},\beta_{2},\ldots,\beta_{K})^{\intercal}$ the full vectors of type I and type II error allocations. The canonical drift parameter, denoted by $\theta$, is the mean of the test statistic at the final analysis under $H_{1}$, that is, $\theta=\delta^{*}\sqrt{I_{K}}$, so that $\mu_{k}=\theta\sqrt{t_{k}}$ for every $k$. It summarises the total information the design must accumulate to attain the target power. In the corresponding fixed-sample design it takes the value $\theta=z_{1-\alpha}+z_{1-\beta}$, and in a GSD the maximum sample size is inflated in proportion to its square,
\begin{linenomath}
	\begin{equation*}
		N_{K}
		=
		N_{0}\frac{\theta^{2}}{(z_{1-\alpha}+z_{1-\beta})^{2}}.
	\end{equation*}
\end{linenomath}
This is the leading factor appearing in the objective functions below, and its ratio to $N_{0}$ is the MIF. For a given schedule $\boldsymbol{t}_{1:K}=(t_{1},t_{2},\ldots,t_{K})^{\intercal}$ and given error allocations $\boldsymbol{\alpha}_{1:K}$ and $\boldsymbol{\beta}_{1:K}$, the value of $\theta$ is uniquely determined as the one at which the design attains power $1-\beta$. It can be computed by the recursive numerical integration of \citet{armitage1969} or, as here, by multivariate normal integration. 
For any positive integer $k \leq K$, let
\begin{linenomath}
	\begin{equation*}
		P_{k}
		=
		\int_{l_{1}}^{u_{1}} \int_{l_{2}}^{u_{2}} \cdots \int_{l_{k}}^{u_{k}} f_{k}(z_{1},z_{2},\ldots,z_{k} \mid \boldsymbol{\mu}_{1:k},\boldsymbol{\Sigma}_{k \times k}) \,\mathrm{d}z_{k} \cdots \mathrm{d}z_{2}\mathrm{d}z_{1}
	\end{equation*}
\end{linenomath}
denote the probability under $H_{1}$ that the trial is still continuing after the $k$-th analysis, that is, that $Z_{k'} \in \mathcal{C}_{k'}$ for every $k' \leq k$. Here $f_{k}(\,\cdot\, \mid \boldsymbol{\mu}_{1:k},\boldsymbol{\Sigma}_{k \times k})$ is the density of the $k$-variate normal distribution with mean vector $\boldsymbol{\mu}_{1:k}$ and covariance matrix $\boldsymbol{\Sigma}_{k \times k}$. For fixed spending functions, both $\theta$ and $\mathcal{C}_{k}$ depend only on $\boldsymbol{t}_{1:K-1}$, and hence so does $P_{k}$.

Suppose the spending functions are prespecified and ensure control of the type I and type II error rates. As shown in \citet{he2025}, the minimisation of the ESS under $H_{1}$ is then restricted to one over the information rates $\boldsymbol{t}_{1:K-1}$ alone, formulated as
\begin{linenomath}
	\begin{equation}
		\label{eqn:2201}
		\begin{aligned}
      & \underset{\boldsymbol{t}_{1:K-1}}{\text{minimise}}
			& & \frac{N_{0}\theta^{2}}{(z_{1-\alpha}+z_{1-\beta})^{2}} \biggl(t_{1} + \sum_{k=2}^{K}(t_{k}-t_{k-1})P_{k-1} \biggr) \\
			& \text{subject to}
			& & 0<t_{1}<t_{2}<\cdots<t_{K-1}<1.
		\end{aligned}
	\end{equation}
\end{linenomath}
A detailed derivation of Eq.~(\ref{eqn:2201}) is provided in \citet{he2025}.

In unrestricted joint optimisation, the search space grows rapidly with the number of stages, and the optima of \citet{eales1992}, \citet{brittain1993} and \citet{jennison2006} covered at most six analyses. To reduce the search space while retaining flexibility in the shape of the stopping boundaries, we parameterise the stopping boundaries using one-parameter spending function families \citep[\textit{e.g.},][]{kim1987,hwang1990,anderson2007}, with a second shape parameter when a futility boundary is included. We consider the HSD family \citep{hwang1990}, defined through the alpha-spending function
\begin{linenomath}
	\begin{equation}
		\label{eqn:2202}
		\alpha(t;\gamma)
		=
		\begin{cases}
			\alpha \cdot \dfrac{1-e^{-\gamma t}}{1-e^{-\gamma}} & \text{if } \gamma \neq 0, \\[6pt]
			\alpha \cdot t & \text{if } \gamma = 0,
		\end{cases}
	\end{equation}
\end{linenomath}
for $t \in [0,1]$. The shape parameter $\gamma$ controls how type I error is distributed across analyses. Negative values of $\gamma$ yield conservative early spending, resembling the OBF design \citep{obrien1979}, whereas positive values yield aggressive early spending, resembling the Pocock design \citep{pocock1977}. The two conventional shapes correspond to particular values of $\gamma$. Matching the Lan--DeMets alpha-spending functions \citep{lan1983} in supremum norm over $t \in (0,1]$ gives $\gamma=-4$ for the OBF shape and $\gamma=1$ for the Pocock shape. The maximum discrepancies in cumulative type I error are $0.0016$ and $0.0002$ respectively at $\alpha=0.025$. The Pocock match is invariant to $\alpha$, since $\alpha$ enters both spending functions as a multiplicative factor, whereas the OBF match depends on $\alpha$ through $z_{1-\alpha/2}$ and becomes more conservative as $\alpha$ decreases. When a futility stopping boundary is incorporated, the same parametric form is used for the beta-spending function, with corresponding shape parameter $\gamma^{*}$. We also consider the KD family \citep{kim1987}, defined as 
\begin{linenomath}
	\begin{equation*}
    \alpha(t;\rho)=\alpha \cdot t^{\rho}
	\end{equation*}
\end{linenomath}
for $t \in [0,1]$ and $\rho>0$, where $\rho \geq 1$ corresponds to conservative early spending and $\rho<1$ to aggressive early spending. The corresponding matches to the conventional shapes are $\rho=3.44$ for OBF and $\rho=0.74$ for Pocock, with maximum discrepancies of $0.0008$ and $0.0007$ at $\alpha=0.025$, and again only the OBF value depends on $\alpha$. Under a given parametric spending function family, $\boldsymbol{\alpha}_{1:K}$ is fully determined by $\boldsymbol{t}_{1:K}$ with the shape parameter $\gamma$ (or $\rho$), and similarly for $\boldsymbol{\beta}_{1:K}$ when a futility stopping boundary is included.

We impose a convex error-spending constraint, motivated by the principle that analyses based on less accumulated information should not receive disproportionately large error allocations. Operationally, this is achieved by requiring the alpha-spending function $\alpha(t;\,\cdot\,)$ and the beta-spending function $\beta(t;\,\cdot\,)$ to be convex in $t$. Convexity means that the marginal rate of error spending is non-decreasing in information time, so that the error spent per unit of information, $(\alpha(t_{k})-\alpha(t_{k-1}))/(t_{k}-t_{k-1})$, is non-decreasing in $k$ for any spacing of the information rates. It does not constrain the raw stagewise increments, which also depend on the interval widths and may decrease when a long interval is followed by a short one. The constraint therefore excludes aggressive early spending from the feasible region without restricting the schedule itself. For the beta-spending function the rationale is weaker than for alpha-spending, since an early futility stop does not inflate the type I error rate. We adopt it as a design convention that treats the two boundaries in the same way. It is not required for error control. 

We let $\boldsymbol{\phi}$ be the boundary shape parameters, where $\boldsymbol{\phi}=\gamma$ or $\boldsymbol{\phi}=(\gamma,\gamma^{*})^{\intercal}$ under the HSD family and $\boldsymbol{\phi}=\rho$ or $\boldsymbol{\phi}=(\rho,\rho^{*})^{\intercal}$ under the KD family, with $\rho^{*}$ the shape parameter of the KD beta-spending function. Convexity requires $\gamma \leq 0$ for the HSD family and $\rho \geq 1$ for the KD family, so the feasible region retains the OBF analogue ($\gamma=-4$, or $\rho=3.44$) and excludes the Pocock analogue ($\gamma=1$, or $\rho=0.74$). Its boundary, $\gamma=0$ or $\rho=1$, corresponds to uniform alpha-spending. Convexity bounds $\gamma$ and $\gamma^{*}$ only from above and $\rho$ and $\rho^{*}$ only from below, so we also impose the search limits $\gamma_{\min}=\gamma^{*}_{\min}=-6$ and $\rho_{\max}=\rho^{*}_{\max}=6$. The feasible region of $\boldsymbol{\phi}$ is then a box, $\Phi = [\gamma_{\min}, 0]$ or $\Phi = [\gamma_{\min}, 0] \times [\gamma^{*}_{\min}, 0]$ under the HSD family, and $\Phi = [1, \rho_{\max}]$ or $\Phi = [1, \rho_{\max}] \times [1, \rho^{*}_{\max}]$ under the KD family, depending on whether a futility boundary is included. The joint optimisation problem can then be formulated as
\begin{linenomath}
	\begin{equation}
		\label{eqn:2203}
		\begin{aligned}
			& \underset{\boldsymbol{t}_{1:K-1},\,\boldsymbol{\phi}}{\text{minimise}}
      & & \frac{N_{0}\theta^{2}}{(z_{1-\alpha}+z_{1-\beta})^{2}} \biggl(t_{1} + \sum_{k=2}^{K}(t_{k}-t_{k-1})P_{k-1} \biggr) \\
			& \text{subject to}
			& & 0<t_{1}<t_{2}<\cdots<t_{K-1}<1 \\
			&
			& & \boldsymbol{\phi} \in \Phi,
		\end{aligned}
	\end{equation}
\end{linenomath}
where both $\theta$ and $P_{k-1}$ are functions of $\boldsymbol{t}_{1:K-1}$ and $\boldsymbol{\phi}$. The box constraints $\boldsymbol{\phi} \in \Phi$ can be handled within the optimiser directly, without the need for penalty terms or inequality checks. Eq.~(\ref{eqn:2203}) minimises the ESS under $H_{1}$ alone and places no constraint on the MSS, which is used only to choose among designs whose ESS is practically equal. We say that one design Pareto-dominates another at a common stage number if neither its EIF under $H_{1}$ nor its MIF is larger and at least one is smaller.

Two restrictions separate Eq.~(\ref{eqn:2203}) from the fully unconstrained problem \citep{eales1992,brittain1993,jennison2006}, in which the critical values $\boldsymbol{u}_{1:K}$ and $\boldsymbol{l}_{1:K}$ are themselves free. The first restricts the boundary shape to a parametric spending function family, replacing $K-1$ free critical values by a single shape parameter, or by two when futility stopping is included. The second is the convex error-spending constraint, which confines $\boldsymbol{\phi}$ to the half of that family for which cumulative spending is convex. The designs reported below are accordingly optimal conditional on the chosen family, the error rates and the number of analyses. 

As a consequence of the derivation of Eq.~(\ref{eqn:2201}) in \citet{he2025}, the joint optimum defined by Eq.~(\ref{eqn:2203}) is invariant to the scale of the design problem, in the following sense.

\begin{prop}
	\label{prop:2201}
	Fix the type I and type II error rates $\alpha$ and $\beta$, the stage number $K$, the futility specification (none, binding or non-binding), and the spending function family together with its feasible region $\Phi$. Under the canonical joint distribution of $\boldsymbol{Z}_{1:K}$, the joint optimum of Eq.~(\ref{eqn:2203}) satisfies
	\begin{linenomath}
		\begin{equation*}
			(\hat{\boldsymbol{t}}_{1:K-1}, \hat{\boldsymbol{\phi}})
			=
			\argmin_{\substack{0<t_{1}<\cdots<t_{K-1}<1 \\ \boldsymbol{\phi} \in \Phi}} \theta^{2} \biggl(t_{1} + \sum_{k=2}^{K}(t_{k}-t_{k-1})P_{k-1} \biggr)
		\end{equation*}
	\end{linenomath}
	and therefore depends only on $(\alpha,\beta,K)$, the futility specification and the spending function family with its feasible region $\Phi$. In particular, $(\hat{\boldsymbol{t}}_{1:K-1}, \hat{\boldsymbol{\phi}})$ is independent of $\delta^{*}$, $\sigma^{2}$ and the allocation ratio $r$, which enter the ESS only through the multiplicative constant $N_{0}$. When the minimiser is not unique, the statement applies to the set of minimisers, and hence to any optimum selected from it.
\end{prop}

\begin{pf}
	Write the objective of Eq.~(\ref{eqn:2203}) as $N_{0}(z_{1-\alpha}+z_{1-\beta})^{-2} \, G(\boldsymbol{t}_{1:K-1},\boldsymbol{\phi})$, where $G$ denotes the product of $\theta^{2}$ with the bracketed expression. The leading factor $N_{0}(z_{1-\alpha}+z_{1-\beta})^{-2}=(\sigma/\delta^{*})^{2}(1+r)^{2}/r$ is a positive constant that does not involve $(\boldsymbol{t}_{1:K-1},\boldsymbol{\phi})$. Minimising the objective is therefore equivalent to minimising $G$, and it suffices to show that $G$ is a function of $(\boldsymbol{t}_{1:K-1},\boldsymbol{\phi})$ and $(\alpha,\beta,K)$ alone.

	Under a given spending function family, $\boldsymbol{\alpha}_{1:K}$ and $\boldsymbol{\beta}_{1:K}$ are determined by $\boldsymbol{t}_{1:K}$ and $\boldsymbol{\phi}$ together with the nominal rates $\alpha$ and $\beta$. Since $I_{k}/I_{k'}=t_{k}/t_{k'}$, the covariance matrix satisfies $\Sigma_{kk'}=\sqrt{\min(t_{k},t_{k'})/\max(t_{k},t_{k'})}$ and so is a function of $\boldsymbol{t}_{1:K}$ only. The stopping boundaries $\boldsymbol{u}_{1:K}$ and $\boldsymbol{l}_{1:K}$ and the canonical drift $\theta$ are obtained jointly, by calibrating $\boldsymbol{\alpha}_{1:K}$ and $\boldsymbol{\beta}_{1:K}$ against $\boldsymbol{\Sigma}_{K \times K}$ and the mean vector $\mu_{k}=\theta\sqrt{t_{k}}$ derived below, since beta-spending is evaluated under $H_{1}$. Here $\theta$ is the value at which the design attains power $1-\beta$, and it is determined by $\boldsymbol{t}_{1:K}$, $\boldsymbol{\alpha}_{1:K}$ and $\boldsymbol{\beta}_{1:K}$ as described earlier in this section. Hence $\boldsymbol{u}_{1:K}$, $\boldsymbol{l}_{1:K}$, the continuation regions $\mathcal{C}_{k}$ and $\theta$ are all functions of $(\boldsymbol{t}_{1:K-1},\boldsymbol{\phi})$ and $(\alpha,\beta,K)$ alone.

	It remains to treat the mean vector. By definition of the canonical drift, $\theta=\delta^{*}\sqrt{I_{K}}$ under $H_{1}$, so that for every $k$,
	\begin{linenomath}
		\begin{equation*}
			\mu_{k}
			=
			\delta^{*}\sqrt{I_{k}}
			=
			\delta^{*}\sqrt{I_{K}} \, \sqrt{I_{k}/I_{K}}
			=
			\theta\sqrt{t_{k}},
		\end{equation*}
	\end{linenomath}
	which involves neither $\delta^{*}$ nor $\sigma$ separately. Every $P_{m}$ appearing in $G$ is therefore a function of $(\boldsymbol{t}_{1:K-1},\boldsymbol{\phi})$ and $(\alpha,\beta,K)$, and so is $G$ itself. Its minimiser over the feasible region inherits the same property.
\end{pf}

\noindent
Proposition~\ref{prop:2201} extends the corresponding result of \citet{he2025} from scheduling-only optimisation to the joint optimisation of interim timing and stopping boundaries. Because the two-step criterion and the cross-$K$ refinement of Section~\ref{sec:24} use only the EIF and MIF, which do not involve $N_{0}$, the reported designs inherit the same invariance. A numerical verification across continuous, binary, count and time-to-event endpoints is reported in Supplemental Material, Section~\href{run:./ZH2023_Supplemental_Material.pdf}{S1}.

\subsection{Stage merging and automatic selection of effective stages}
\label{sec:23}
The joint optimisation problem in Eq.~(\ref{eqn:2203}) takes the stage number $K$ as a fixed input, leaving open the practical question of how $K$ should be determined. Two distinct considerations bear on that choice, and we address them through two separate devices. The first is that additional analyses yield diminishing reductions in the ESS while continuing to commit a larger MSS and a further round of data monitoring, so that beyond some number of analyses a further analysis is no longer justified. This consideration is handled by an automatic selection of the effective number of analyses, governed by a tolerance on the ESS. The second is that nothing in the formulation of Eq.~(\ref{eqn:2203}) prevents the optimiser, when $K$ is larger than the problem effectively requires, from placing adjacent information rates so close together that the corresponding interim looks become operationally indistinguishable. This consideration is handled by a stage-merging rule. Imposing a penalty on the minimum gap between adjacent information rates to prevent looks from collapsing in this way would move the solution away from the optimum of the unpenalised problem. It would also hide potentially useful information about the appropriate number of analyses. We thus treat collapse, where it arises, as informative and address it after the optimisation rather than as a constraint within it. As we solve Eq.~(\ref{eqn:2203}) at every nominal $K$ from 2 to a specified maximum, resolving a collapse amounts to retrieving the optimum already computed at a smaller $K$, so no new optimisation is needed.

Let $(\hat{\boldsymbol{t}}_{1:K-1}^{(K)}, \hat{\boldsymbol{\phi}}^{(K)})$ be the joint optimum obtained at nominal $K$, with $\hat{t}_{0}^{(K)}=0$ and $\hat{t}_{K}^{(K)}=1$ by convention. We define the between-analysis information gaps, with the first measured from information zero, as $\Delta t_{k}^{(K)}=\hat{t}_{k}^{(K)}-\hat{t}_{k-1}^{(K)}$ for $k=1,2,\ldots,K$. Given a prespecified merging threshold $\epsilon$, we say that the design collapses if $\min_{k} \Delta t_{k}^{(K)}<\epsilon$. In that case, we define the effective stage number $K_{\text{eff}}^{(K)}$ by passing in turn to the optima of Eq.~(\ref{eqn:2203}) at $K-1, K-2, \ldots$ until the resulting optimum satisfies $\min_{k} \Delta t_{k}\geq\epsilon$. If the design does not collapse, $K_{\text{eff}}^{(K)}=K$. The merging threshold is specified on the information rate scale, which ensures that the same $\epsilon$ is meaningful across trials of different sizes. A default value of $\epsilon=0.05$ requires successive analyses to be separated by at least 5\% of the total information. The choice of threshold may be adapted to the operational context. What makes two looks operationally distinguishable depends on the rate at which information accrues. For non-survival endpoints this is the recruitment rate, and for survival endpoints the event accrual rate. It also depends on the fixed cost of convening a data monitoring committee and preparing an interim report. 

Based on the stage-merging rule, we introduce an automatic selection of the effective number of analyses. For each nominal $K$ between 2 and a specified maximum $K_{\max}$, we compute the merged solution, with the effective stage number denoted $K_{\text{eff}}^{(K)}$, and the corresponding ESS under $H_{1}$, denoted $\mathbb{E}(N^{(K)} \mid H_{1})$. Let $\mathbb{E}(N^{\star} \mid H_{1})=\min_{K} \mathbb{E}(N^{(K)} \mid H_{1})$ represent the smallest ESS achievable across all nominal $K$. The recommended effective stage number $K^{\star}$ can then be defined as
\begin{linenomath}
	\begin{equation}
		\label{eqn:2301}
		K^{\star}
		=
		\min\bigl\{K_{\text{eff}}^{(K)}: \mathbb{E}(N^{(K)} \mid H_{1}) \leq (1+\tau) \cdot \mathbb{E}(N^{\star} \mid H_{1})\bigr\},
	\end{equation}
\end{linenomath}
where $\tau$ is a prespecified tolerance parameter with default value $\tau=0.01$, corresponding to a 1\% ESS tolerance. That is, $K^{\star}$ is the effective stage number of the smallest nominal $K$ whose ESS under $H_{1}$ is within $\tau$ of the minimum achievable value. In all configurations examined with $K \le 7$, the MSS increased with $K$, and the refinement step described in Section~\href{run:./ZH2023_Supplemental_Material.pdf}{S2} requires it not to decrease. The smallest nominal $K$ whose ESS was within the tolerance therefore also had the smallest MSS among the qualifying candidates. We do not claim this as a general property of the canonical joint distribution, as adding an analysis enlarges the feasible set without constraining the MSS of the ESS-minimising solution. 

\subsection{\texttt{OptimInterim}}
\label{sec:24}
We extend \texttt{OptimInterim}, the \texttt{R}-based tool of \citet{he2025} for scheduling-only optimisation, to jointly optimise the timing of interim analyses and the shape of the stopping boundaries in GSDs. The aim is to minimise the ESS under $H_{1}$ while controlling overall type I and type II error rates. The extended \texttt{OptimInterim} retains the flexibility of the original version, supporting both one- and two-sided tests and allowing either binding or non-binding futility stopping rules. The joint optimisation is implemented for both the HSD family of \citet{hwang1990} and the KD family of \citet{kim1987}.

The interim timing is parameterised through a softmax transformation of the between-analysis spacings. The joint minimisation over that parameterisation and $\boldsymbol{\phi}$ is performed by a two-phase scheme, a global controlled random search with local mutation \citep{kaelo2006} followed by a BOBYQA polish \citep{powell2009} from a pool of diversified starts. Operating characteristics are computed by multivariate normal integration \citep{genz1992,genz2024} rather than by simulation. Full details of the parameterisation, the optimiser settings and the integration tolerances are given in Section~\href{run:./ZH2023_Supplemental_Material.pdf}{S2}. The choice of BOBYQA for the polish step is supported by a benchmark comparison with nine alternative optimisers, reported in Section~\href{run:./ZH2023_Supplemental_Material.pdf}{S3}. The operating characteristics are cross-validated against Monte Carlo simulation in Section~\href{run:./ZH2023_Supplemental_Material.pdf}{S4}.

Because Eq.~(\ref{eqn:2203}) is solved separately at each nominal $K$ and its objective is nearly flat near the optimum, solutions at neighbouring $K$ need not be mutually consistent, so a refinement step is applied across $K$. The designs reported in Section~\ref{sec:3} are therefore the outcome of a two-step criterion: the ESS under $H_{1}$ is minimised to within a tolerance of $10^{-3}$ in EIF, and the MSS is then minimised among the designs attaining that tolerance. The cross-$K$ requirements imposed by the refinement, and the basis for this tolerance, are given in Section~\href{run:./ZH2023_Supplemental_Material.pdf}{S2}.

\subsection{Performance evaluation}
\label{sec:25}
Except in the case study, we considered two-arm superiority trials under equal randomisation, comparing a treatment with a control, and used the \texttt{rpact} package \citep{wassmer2024} to compute the stopping boundaries and to translate designs into sample sizes and event counts. The operating characteristics of interest are the MSS and the ESSs under $H_{0}$, $H_{0}/H_{1}$ and $H_{1}$, where $H_{0}/H_{1}$ represents a target effect size of $\delta=\delta^{*}/2$. As in \citet{he2025}, the inflation factor is independent of the standardised effect size, so we instead present the EIFs under $H_{0}$, $H_{0}/H_{1}$ and $H_{1}$, together with the MIF, as defined in Section~\ref{sec:1}.

\subsubsection{Operating characteristics evaluation}
\label{sec:251}
We evaluated the operating characteristics of jointly optimised GSDs across a range of stage numbers, error rates and stopping boundaries. We considered a continuous primary endpoint and assumed that each participant’s response is independent and identically normally distributed with known common standard deviation $\sigma=1$. The null hypothesis $H_{0}:\delta \leq 0$ is tested against the alternative hypothesis $H_{1}:\delta>0$ with assumed effect size of $\delta^{*}=0.5$. We fixed the one-sided type I error rate at $\alpha=0.025$ and varied the type II error rate over $\beta \in \{0.1,0.2\}$, corresponding to 90\% and 80\% power, respectively.

For each combination of type I and type II error rates, we evaluated jointly optimised GSDs under both the HSD family of \citet{hwang1990} and the KD family of \citet{kim1987}. We initially focused on efficacy-only designs, which are the more common starting point in confirmatory trial planning and the simpler setting in which to isolate the effect of joint optimisation on the ESS. For each configuration, we considered cases with up to six interim looks, corresponding to a total of seven stages including the final analysis, generated by the extended \texttt{OptimInterim}. 

We compared the jointly optimal designs against two scheduling-only optimal GSDs from \citet{he2025}, the optimal OBF design and the optimal Pocock design. In each case, the stopping boundary is implemented with the Lan--DeMets alpha-spending approximation \citep{lan1983}, and only the information rates are optimised. The two reference designs represent the scheduling-only counterparts of the two most widely used conventional boundary shapes, isolating the incremental efficiency gain from jointly optimising boundary shape in addition to interim timing.

For each $K$ and each spending function family, we report the jointly optimal information rates $\hat{\boldsymbol{t}}_{1:K-1}$ and shape parameters. The efficacy shape is denoted $\hat{\gamma}$ under the HSD family and $\hat{\rho}$ under the KD family, with $\hat{\gamma}^{*}$ or $\hat{\rho}^{*}$ for futility where applicable. We also report the effective number of stages $K_{\text{eff}}$ after applying the stage-merging rule at $\epsilon=0.05$, and the recommended $K^{\star}$ from Eq.~(\ref{eqn:2301}) at $\tau=0.01$. We examined their sensitivity to $\epsilon \in \{0.03, 0.10\}$, roughly half and double the default, and to $\tau \in \{0.005, 0.02, 0.05\}$. We report the MIF and EIFs under $H_{0}$, $H_{0}/H_{1}$ and $H_{1}$ for each jointly optimal design and for the two reference designs, to quantify both the additional ESS saving and any associated MSS inflation. To assess how $K^{\star}$ depends on the range searched, we further computed the efficacy-only optima at $K=8$ under both families and both power levels, applied the stage-merging rule to them and recomputed $K^{\star}$ over $K \le 8$. 

\subsubsection{Benchmark against unrestricted joint optima}
\label{sec:252}
We benchmarked the HSD family joint optima, with and without the convex alpha-spending constraint, against the unrestricted joint optima of \citet{brittain1993}, obtained by simultaneous optimisation over information rates and critical values with no restriction to a parametric boundary family. These unrestricted optima were taken from their published results rather than recomputed here. The unrestricted optima of \citet{eales1992} and \citet{jennison2006} were not used as benchmarks, because their tests have binding futility boundaries and other error rates or optimality criteria. We considered $K \in \{2, 3\}$ with $(\alpha,\beta) \in \{(0.025, 0.1), (0.025, 0.2)\}$ and efficacy-only stopping, the configurations for which \citet{brittain1993} reported numerical optima over both critical values and interim timings. Without the constraint, the parametric optimum is unconstrained in the shape parameter alone. We then compared both parametric optima with the unrestricted numerical optima of \citet{brittain1993}, reporting the EIF under $H_{1}$.

\subsubsection{Scaling analysis}
\label{sec:253}
Adding interim analyses lowers the EIF with diminishing returns while raising the MIF, and this trade-off underlies the choice of $K^{\star}$ in Section~\ref{sec:23}. To describe each trend with a few interpretable parameters, and to compare designs, we fitted simple curves in $K$ to the EIF and the MIF of the jointly optimal design. At $(\alpha,\beta) \in \{(0.025, 0.1), (0.025, 0.2)\}$, we ran the joint optimisation for $K \in \{2, 3, \ldots, 7\}$ under the HSD and KD families with efficacy stopping. We fitted the power law $\mathrm{EIF}_{K} \approx a + bK^{-p}$ and the analogous law $\mathrm{MIF}_{K} \approx a' + b'K^{q}$ by nonlinear least squares. In the EIF law, $a$ is the limit approached as $K$ grows and $p$ is the rate of approach, which is one in the discretisation theory of optimal stopping for a one-dimensional diffusion \citep{dupuis2005}. In the MIF law, $q<0$ with $b'<0$ corresponds to an MIF that rises towards a ceiling $a'$. We report $\hat{p}$ and $\hat{q}$ with their standard errors and goodness-of-fit summaries. For comparison, we fitted the same laws to the scheduling-only optimal OBF and Pocock designs of \citet{he2025}. 

To visualise this trade-off, we plot the $(\mathrm{EIF}_{K}, \mathrm{MIF}_{K})$ frontier with points labelled by $K$, so that the marginal ESS saving and marginal MSS cost of each additional interim analysis can be compared. This comparison illustrates the rationale for the $K^{\star}$ selection rule described in Section~\ref{sec:23}. 

\subsection{Case study: the CORRECT trial}
\label{sec:26}
We show the practical use of the proposed approach through a redesign of the CORRECT study \citep{grothey2013}, a Phase III survival trial of regorafenib in previously treated metastatic colorectal cancer. The survival-endpoint setting demonstrates the asymptotic applicability of the joint optimisation framework beyond the continuous-endpoint setting in which the framework is developed. It is also the most operationally demanding setting. The information-time schedule on the event-count scale must be translated into a calendar-time schedule through an accrual and survival model, and the resulting sample size and trial duration respond differently to the design choices than in the continuous and binary cases. CORRECT also randomised patients unequally between arms, so the redesign relies on the invariance of the joint optimum to the allocation ratio $r$ established in Proposition~\ref{prop:2201}. The redesign compared the jointly optimal HSD family design with scheduling-only optimal OBF and Pocock designs on both the event-count scale and the calendar-time scale, quantifying the efficiency gains attributable to joint optimisation of boundary shape and interim timing.

\subsubsection{Trial background and redesign}
\label{sec:261}
CORRECT was an international, multicentre, placebo-controlled, double-blind, randomised Phase III study comparing regorafenib to placebo, both plus best supportive care. The primary endpoint was overall survival (OS), and patients were randomised in a 2:1 ratio, so that $r=2$. The design assumed a target hazard ratio of $\psi^{*}=0.75$, with 90\% power at a one-sided type I error rate of 0.025 and an assumed control-arm median OS of 4.5 months. It required 582 deaths for the final analysis, from a planned enrolment of about 690 patients. Interim monitoring of OS adopted the Lan--DeMets alpha spending function with an OBF boundary shape, and the futility boundaries were calculated independently of the efficacy boundary, so futility stopping was non-binding.

Two interim analyses were planned, the first at approximately 174 events (30\% information) for futility only, and the second at approximately 408 events (70\% information) for both efficacy and futility \citep{fda2012}. In practice, they were conducted at 301 events (51.7\%) and 432 events (74.2\%), and the trial stopped for efficacy at the second. Neither schedule was chosen on efficiency grounds, and CORRECT is in this respect representative of the practice in which interim timing is inherited from administrative and reporting constraints. As OS was the primary endpoint, its schedule was not tied to that of another endpoint and was in principle free to be set on statistical grounds. The redesign is thus a design-level counterfactual.

We redesigned the OS endpoint under a configuration allowing up to $K_{\max}=7$ total stages, corresponding to at most six interim analyses. The stage-merging rule and the selection criterion of Section~\ref{sec:23} were applied at their defaults, $\epsilon=0.05$ and $\tau=0.01$, so that successive analyses had to be at least 5\% of the maximum number of events apart. All trial parameters other than the shape of stopping boundaries, the number and timing of interim analyses, and the resulting maximum numbers of events and subjects were kept identical to the original. We compared the jointly optimal HSD family design against the scheduling-only optimal OBF and Pocock designs of \citet{he2025}, with the original CORRECT configuration included as a historical reference design at its realised information times. The redesign candidates used efficacy stopping only, so the reference design was evaluated on its efficacy schedule alone, with a single efficacy interim analysis at 432 events (74.2\%) before the final analysis, to keep the comparison like for like. The original futility rule was specified as a threshold on the observed hazard ratio rather than through a beta-spending function. Its operating characteristics are reported separately, because a design that can stop for futility is not comparable with efficacy-only designs in expected sample size under $H_{0}$.

\subsubsection{Calendar-time analysis at the CORRECT baseline}
\label{sec:262}
On the event-count scale, the operating characteristics are determined by the information-time schedule, the spending function shape and the fixed-design number of events. Because $I_{k}=d_{k}r/(1+r)^{2}$ involves neither the accrual pattern nor the baseline hazard, neither the schedule nor the shape depends on them, and neither do the event-count operating characteristics. However, the operationally relevant quantities, the ESS $\mathbb{E}(N \mid H_{h})$ and the expected trial duration (ETD) $\mathbb{E}(T \mid H_{h})$, with $H_{h}$ denoting $H_{0}$, $H_{0}/H_{1}$ or $H_{1}$, depend on these additional design inputs, which we now specify at the CORRECT baseline.

We adopted the planned accrual model of the original protocol, namely 30 patients per month following an initial four-month ramp-up period, together with a projected drop-out rate of 3\% per year \citep{fda2012}. The ramp-up was represented as a piecewise-constant accrual intensity of 15 patients per month over the first four months. Analyses were event-driven, with each triggered when the corresponding target event count was reached and administrative censoring applied at the time of the final analysis. We adopted exponential survival with a control-arm median OS of 4.5 months. Each candidate design was sized separately to its own maximum number of events $d_{K}$ under the standard size-for-events convention, and the accrual duration was chosen so that the maximum number of subjects $N_{K}$ scaled linearly with $d_{K}$. The scaling was set by the original CORRECT design, which had $N_{K}=690$ over a 25-month accrual period. For each candidate design, we report $\mathbb{E}(N \mid H_{h})$ and $\mathbb{E}(T \mid H_{h})$ under $H_{0}$, $H_{0}/H_{1}$ and $H_{1}$, together with the MSS $N_{K}$ and the maximum trial duration (MTD) $T_{K}$, the calendar time of the final analysis under $H_{1}$. The event-accrual curve was computed with the formulas of \citet{rubinstein1981} and \citet{lachin1986} for exponential survival.

\section{Results}
\label{sec:3}
Section~\ref{sec:31} reports the jointly optimal efficacy-only designs and their operating characteristics, which serve as reference tables for every endpoint type considered in Section~\ref{sec:21}. Section~\ref{sec:32} benchmarks these designs against the unrestricted optima of \citet{brittain1993}, Section~\ref{sec:33} examines how their efficiency scales with the number of analyses, and Section~\ref{sec:34} applies them to the CORRECT trial.

\subsection{Operating characteristics}
\label{sec:31}
Table~\ref{tab:311} reports the jointly optimal information rates and shape parameters under the HSD and KD families for up to $K=7$ stages, at a one-sided type I error rate of 0.025 and type II error rates of 0.1 and 0.2. It also reports the effective stage number $K_{\text{eff}}$ after stage merging and the recommended $K^{\star}$. Figure~\ref{fig:311} shows the corresponding MIFs and EIFs under $H_{0}$, $H_{0}/H_{1}$ and $H_{1}$, alongside those of the scheduling-only optimal OBF and Pocock designs of \citet{he2025}. Results under GSDs incorporating both efficacy and non-binding futility stopping boundaries are reported in Section~\href{run:./ZH2023_Supplemental_Material.pdf}{S5}.

\begin{table}[!ht]
	\centering
	\begin{tabular}{llcccccc}
	\hline
	Family & $\beta$ & $K$ & $\hat{t}_{1:K-1}$ (\%) & $\hat{\phi}$ & $K_{\text{eff}}$ & MIF & EIF$(H_1)$ \\
	\hline
	HSD & 0.1 & 2 & 51.0 & 0.000 & 2 & 1.075 & 0.777 \\
	 &  & 3 & 37.7/65.7 & 0.000 & 3 & 1.106 & 0.719 \\
	 &  & 4 & 31.5/51.4/72.7 & 0.000 & 4 & 1.122 & 0.692 \\
	 &  & 5 & 27.8/44.5/61.4/77.6 & 0.000 & 5 & 1.134 & 0.678 \\
	 &  & \textbf{6} & 24.8/38.7/53.3/65.5/80.0 & 0.000 & 6 & 1.140 & 0.669 \\
	 &  & 7 & 23.4/35.0/47.0/55.8/67.8/81.8 & 0.000 & 7 & 1.145 & 0.663 \\
	\cline{2-8}
	 & 0.2 & 2 & 55.2 & -0.538 & 2 & 1.068 & 0.848 \\
	 &  & 3 & 43.2/70.3 & -0.835 & 3 & 1.084 & 0.806 \\
	 &  & 4 & 36.4/56.0/76.1 & -0.860 & 4 & 1.097 & 0.786 \\
	 &  & 5 & 31.9/49.2/66.0/80.2 & -0.879 & 5 & 1.105 & 0.775 \\
	 &  & \textbf{6} & 28.6/43.0/58.7/69.6/82.7 & -0.887 & 6 & 1.111 & 0.769 \\
	 &  & 7 & 27.6/38.8/51.5/58.1/69.9/83.4 & -0.905 & 7 & 1.113 & 0.765 \\
	\hline
	KD & 0.1 & 2 & 51.0 & 1.000 & 2 & 1.075 & 0.777 \\
	 &  & 3 & 37.7/65.7 & 1.000 & 3 & 1.106 & 0.719 \\
	 &  & 4 & 31.5/51.4/72.7 & 1.000 & 4 & 1.122 & 0.693 \\
	 &  & 5 & 27.7/44.4/61.5/77.5 & 1.000 & 5 & 1.134 & 0.678 \\
	 &  & \textbf{6} & 24.9/38.7/53.5/65.8/80.1 & 1.000 & 6 & 1.141 & 0.669 \\
	 &  & 7 & 23.4/35.0/47.0/55.8/67.8/81.8 & 1.000 & 7 & 1.145 & 0.663 \\
	\cline{2-8}
	 & 0.2 & 2 & 55.4 & 1.244 & 2 & 1.066 & 0.848 \\
	 &  & 3 & 42.7/69.3 & 1.278 & 3 & 1.090 & 0.806 \\
	 &  & 4 & 36.6/55.7/76.2 & 1.377 & 4 & 1.096 & 0.786 \\
	 &  & 5 & 31.9/48.6/65.4/79.6 & 1.397 & 5 & 1.104 & 0.775 \\
	 &  & \textbf{6} & 27.9/42.1/57.6/68.7/81.7 & 1.412 & 6 & 1.108 & 0.768 \\
	 &  & 7 & 26.8/37.9/50.6/57.6/68.8/82.3 & 1.435 & 7 & 1.108 & 0.764 \\
	\hline
	\end{tabular}
	\caption{Jointly optimal information rates (\%) and shape parameters under the HSD family and the KD family for efficacy-only GSDs incorporating up to six interim analyses with a one-sided type I error rate of 0.025. The shape parameter column $\hat{\phi}$ corresponds to $\hat{\gamma}$ under the HSD family and to $\hat{\rho}$ under the KD family. The recommended $K^{\star}$ under each combination of family and power level is highlighted in bold. The columns $K_{\text{eff}}$, MIF and EIF$(H_{1})$ give the effective number of stages after stage merging at $\epsilon=0.05$, the maximum inflation factor and the expected inflation factor under $H_{1}$.}
	\label{tab:311}
\end{table}

\begin{figure}[!ht]
	\centering
	\includegraphics[width=\linewidth]{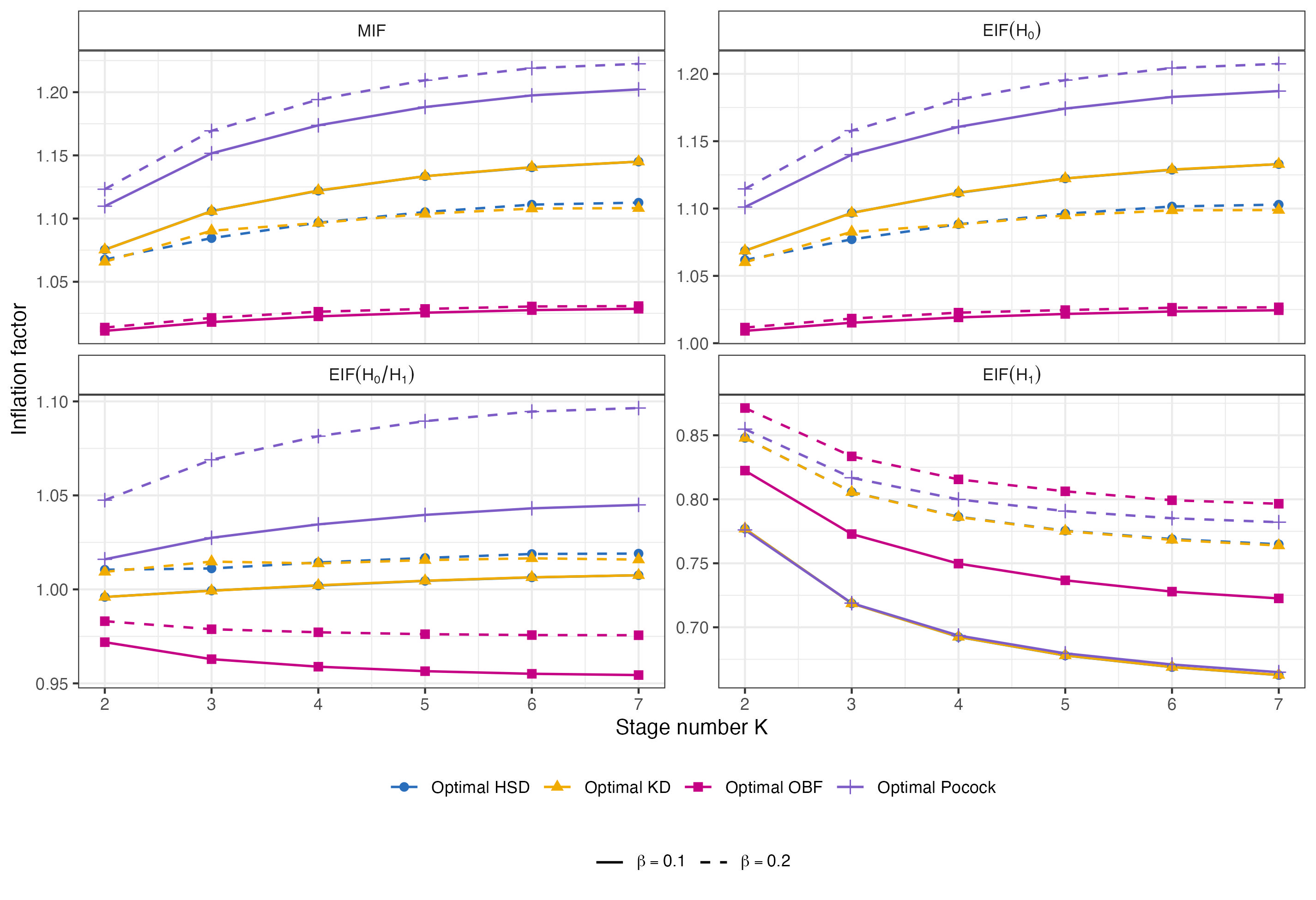}
	\caption{Operating characteristics for jointly optimised GSDs under the HSD family and the KD family, incorporating up to six interim analyses for efficacy, with a one-sided type I error rate of 0.025. The two families are overlaid within each panel. The operating characteristics of the scheduling-only optimal Pocock and OBF GSDs of \citet{he2025} are also shown for comparison.}
	\label{fig:311}
\end{figure}

The two families yielded operating characteristics that were numerically indistinguishable on the EIF and MIF scales. Across all $(K, \beta)$ cells in Table~\ref{tab:311}, the HSD and KD joint optima agreed on EIF under $H_{1}$ to within $0.001$ and on MIF to within $0.006$, despite parametrically distinct shape parameters at $\beta=0.2$. At $\beta=0.1$ both optima lie at uniform spending, where the two families coincide, so the agreement there holds by construction. The jointly optimal information rates also agreed across the two families to within $1.2$ percentage points at each $K$. At $\beta=0.1$, both families pinned the shape parameter at the boundary of the feasible region for all $K$ from 2 to 7 ($\hat{\gamma}=0$ under HSD and $\hat{\rho}=1$ under KD), the point of uniform alpha-spending. At $\beta=0.2$, the optimum was interior, with $\hat{\gamma}$ ranging from $-0.54$ at $K=2$ to $-0.90$ at $K=7$ and $\hat{\rho}$ ranging from $1.24$ to $1.43$, both becoming strictly more conservative as $K$ grows, a direction that the refinement step of Section~\href{run:./ZH2023_Supplemental_Material.pdf}{S2} requires.

The marginal ESS efficiency gain from each additional interim look diminished rapidly.
At $\beta=0.1$, the EIF under $H_{1}$ decreased from $0.777$ at $K=2$ to $0.669$ at $K=6$ and to $0.663$ at $K=7$. At $\beta=0.2$, the EIF under $H_{1}$ fell from $0.848$ at $K=2$ to $0.769$ at $K=6$. At both power levels, $K=6$ was the smallest number of stages whose EIF lay within a factor $(1+\tau)$ of the minimum over the range searched, with $\tau=0.01$. Consequently, the recommended effective number of stages was $K^{\star}=6$ at both power levels under both families. The stage-merging rule did not bind anywhere in Table~\ref{tab:311}, with $K_{\text{eff}}=K$ at every $K$ under both families and both power levels. The jointly optimal schedules therefore contained no operationally indistinguishable looks at the default threshold $\epsilon=0.05$, and $K^{\star}$ was determined by the tolerance criterion of Eq.~(\ref{eqn:2301}) alone. The margins on that criterion were narrow. Under the HSD family at $\beta=0.1$, the EIF at $K=6$ exceeded the smallest value attainable over the range $K\le 7$ by $0.94\%$ and that at $K=5$ exceeded it by $2.31\%$, against the tolerance $\tau=0.01$. At $\beta=0.2$ the corresponding excesses were $0.53\%$ and $1.39\%$. The KD family gave the same excesses at $\beta=0.1$, and $0.58\%$ and $1.48\%$ at $\beta=0.2$. The recommendation was therefore sensitive to $\tau$, and the stage-merging threshold did not reduce this sensitivity (Section~\href{run:./ZH2023_Supplemental_Material.pdf}{S6}).

The criterion is relative to the search range as well, because Eq.~(\ref{eqn:2301}) compares each $K$ with the smallest EIF attainable over the range examined, and widening that range can only lower the reference value. The eight-stage optima of Section~\ref{sec:251} affected the two power levels differently. At $\beta=0.2$ their smallest gaps were $0.048$ under the HSD family and $0.046$ under the KD family, so both merged to seven stages at $\epsilon=0.05$. Even without merging, their EIF under $H_{1}$ of $0.763$ and $0.762$ would have left $K=6$ within the tolerance, so $K^{\star}$ remained $6$. At $\beta=0.1$ the smallest gap was $0.061$ under the HSD family and $0.055$ under the KD family, so neither optimum merged. Their EIF under $H_{1}$, $0.6594$ and $0.6597$, was low enough that $K=6$ no longer met the tolerance, and $K^{\star}$ became $7$. These eight-stage optima had a smaller MSS than the seven-stage optima, so they did not satisfy the cross-$K$ requirements of the refinement step (Section~\href{run:./ZH2023_Supplemental_Material.pdf}{S2}) and serve only to test how $K^{\star}$ depends on the range searched. Extending the range to $K=8$ therefore left the recommendation of six analyses unchanged at 80\% power, whereas at 90\% power the recommendation held conditionally on the range $K \le 7$.

At the recommended $K^{\star}=6$, the jointly optimal design reduced the ESS under $H_{1}$ by about 33\% at 90\% power (EIF$=0.669$) and 23\% at 80\% power (EIF$=0.769$), relative to the fixed-sample design. Under the HSD family, the corresponding MSS inflation was 14.0\% and 11.1\%, respectively (14.1\% and 10.8\% under the KD family). The MIF increased slowly with $K$ under both families. At $\beta=0.1$ it increased from $1.075$ at $K=2$ to approximately $1.145$ at $K=7$, and at $\beta=0.2$ it increased from approximately $1.07$ at $K=2$ to approximately $1.11$ at $K=7$, with the marginal increment from $K=6$ to $K=7$ at most $0.005$ across both families. Figure~\ref{fig:311} also shows the EIFs under $H_{0}$ and $H_{0}/H_{1}$, which behaved similarly across the two families. The EIF under $H_{0}$ stayed close to the MIF throughout, because the probability of early stopping under the null is very low. The EIF under $H_{0}/H_{1}$ remained close to 1, corresponding to the design point at which the ESS is approximately equal to the sample size required in the fixed-sample design.

Compared with the scheduling-only optimal Pocock design of \citet{he2025}, the jointly optimal design required strictly smaller MSS inflation at every $K$ and both power levels. It also achieved at least the same ESS saving under $H_{1}$, except at $K=2$ and $\beta=0.1$, where the convexity constraint pinned $\hat{\gamma}$ at $0$ and optimal Pocock attained a marginally smaller EIF ($0.7761$ versus $0.7770$). Pareto-dominance therefore held at both power levels for $K \geq 3$. At $\beta=0.1$ and the recommended $K^{\star}=6$ the EIFs under $H_{1}$ were nearly identical ($0.669$ versus $0.671$), so the dominance lay mainly in the MIF ($1.140$ versus $1.197$). At $\beta=0.2$ the joint optimum was strictly better on both axes ($0.769$ versus $0.785$ in EIF and $1.111$ versus $1.219$ in MIF). Compared with the scheduling-only optimal OBF design, the jointly optimal design had a smaller EIF under $H_{1}$ at the cost of a larger MIF, placing it at a different point on the efficiency-inflation trade-off frontier. The joint optimum therefore retained ESS efficiency comparable to that of the optimal Pocock design while requiring a substantially smaller MSS. Scheduling-only optimisation could not achieve this combination under either conventional boundary shape. For conciseness, the remainder of Section~\ref{sec:3} reports results under the HSD family only.

\subsection{Benchmark against unrestricted joint optima}
\label{sec:32}
Table~\ref{tab:321} summarises the EIF under $H_{1}$ for three designs. These are the Brittain--Bailey unrestricted joint optimum (free over critical values and information rates) and two HSD family joint optima. One omits the convex alpha-spending constraint, so that $\hat{\gamma}$ can take any sign, and the other is convexity-constrained with $\hat{\gamma}\in[-6,0]$. The Brittain--Bailey EIF values were taken from Table~2 of \citet{brittain1993}, with their two- and three-stage ESS values normalised by the sample size required in the corresponding fixed-sample design. The convexity-constrained values were taken from the same joint optimisation that produced Table~\ref{tab:311}, so the two tables agree at the shared $(K,\beta)$ cells.

\begin{table}[!ht]
	\centering
	\begin{tabular}{ccccccc}
	\hline
	 & & \multicolumn{1}{c}{Brittain--Bailey} &
	\multicolumn{2}{c}{Unconstrained HSD} &
	\multicolumn{2}{c}{Constrained HSD} \\
	\cline{4-5} \cline{6-7}
	$K$ & $\beta$ & EIF$(H_{1})$ &
	EIF$(H_{1})$ & $\hat{\gamma}$ &
	EIF$(H_{1})$ & $\hat{\gamma}$ \\
	\hline
	2 & 0.1 & 0.7759 & 0.7758 & 0.76 & 0.7770 & 0.00 \\
	2 & 0.2 & 0.8479 & 0.8479 & -0.54 & 0.8479 & -0.54 \\
	3 & 0.1 & 0.7182 & 0.7179 & 0.56 & 0.7187 & 0.00 \\
	3 & 0.2 & 0.8065 & 0.8057 & -0.84 & 0.8057 & -0.84 \\
	\hline
	\end{tabular}
	\caption{Comparison of EIFs under $H_{1}$ for the unrestricted joint optimum of \citet{brittain1993}, the unconstrained HSD family joint optimum and the convexity-constrained HSD family joint optimum ($\hat{\gamma}\in[-6,0]$), for $K\in\{2,3\}$ and $(\alpha,\beta)\in\{(0.025,0.1),(0.025,0.2)\}$ under efficacy-only stopping. At $\beta=0.2$ the constrained optimum is interior, and hence also unconstrained, so its values are shown in both columns.}
	\label{tab:321}
\end{table}

Across all four $(K, \beta)$ configurations, the unconstrained HSD family joint optima agreed with the unrestricted Brittain--Bailey optima to within $0.001$ in EIF under $H_{1}$. The unconstrained HSD values lay slightly below the published values in three cells, which exact optima would not allow. \citet{brittain1993} evaluated their optima at a single maximum sample size for each power, for example $1{,}110$ per $1{,}000$ fixed-sample subjects at 90\% power, and described the resulting ESS as essentially minimum, so these small differences reflect the precision of the published values. Restricting the boundary shape to a one-parameter HSD family therefore imposed a negligible efficiency loss relative to free optimisation over critical values. This is consistent with \citet{jennison2006}, who found for tests with binding futility boundaries that simple designs within a monomial spending family attain efficiency close to that of designs with arbitrary group sizes and arbitrary boundaries. The behaviour of the additional convex alpha-spending constraint depended on the power configuration. At $\beta=0.1$ the unconstrained HSD optimum occurred at $\hat{\gamma}\approx 0.76$ ($K=2$) and $\hat{\gamma}\approx 0.56$ ($K=3$), corresponding to mildly Pocock-like early spending. Imposing the convexity constraint pushed $\hat{\gamma}$ to the boundary value $0$ at the cost of a $0.0008$--$0.0012$ increase in EIF. At $\beta=0.2$ the constrained optimum was in the interior of the feasible region ($\hat{\gamma}\approx-0.54$ at $K=2$ and $\hat{\gamma}\approx-0.84$ at $K=3$), so the convexity constraint was non-binding and the constrained optimum was also the unconstrained one. A separate unconstrained search ended elsewhere on the nearly flat ridge in $\gamma$, at the same EIF under $H_{1}$ to four decimal places.

\subsection{Scaling analysis}
\label{sec:33}
Table~\ref{tab:331} summarises the fitted exponents $\hat{p}$ and $\hat{q}$ in the models $\mathrm{EIF}_{K} \approx a + bK^{-p}$ and $\mathrm{MIF}_{K} \approx a' + b'K^{q}$, with nonlinear least-squares standard errors and goodness-of-fit summaries $R^{2}_{\mathrm{EIF}}$ and $R^{2}_{\mathrm{MIF}}$. Results are given for $K \in \{2, \ldots, 7\}$ at $(\alpha,\beta) \in \{(0.025,0.1),(0.025,0.2)\}$ under efficacy-only stopping, for the jointly optimal HSD family design and the scheduling-only optimal OBF and Pocock designs of \citet{he2025}. 
Figure~\ref{fig:331} plots the EIF and MIF curves against $K$ on a log $K$-scale with the fitted power laws overlaid, and Figure~\ref{fig:332} the $(\mathrm{EIF}_{K}, \mathrm{MIF}_{K})$ trade-off frontier for $K$ from 2 to 7. The corresponding scaling analysis under designs with both efficacy and non-binding futility stopping boundaries is reported in Section~\href{run:./ZH2023_Supplemental_Material.pdf}{S5}.

\begin{table}[!ht]
	\centering
	\begin{tabular}{llcccc}
	\hline
	Design & $\beta$ & $\hat{p}$ (SE) & $R^{2}_{\mathrm{EIF}}$ & $\hat{q}$ (SE) & $R^{2}_{\mathrm{MIF}}$ \\
	\hline
	Optimal HSD & 0.1 & 1.31 (0.00) & 1.0000 & -0.80 (0.05) & 0.9997 \\
	Optimal HSD & 0.2 & 1.30 (0.02) & 1.0000 & -0.51 (0.20) & 0.9959 \\
	Optimal OBF & 0.1 & 1.20 (0.02) & 1.0000 & -0.61 (0.10) & 0.9990 \\
	Optimal OBF & 0.2 & 1.28 (0.06) & 0.9997 & -1.01 (0.20) & 0.9966 \\
	Optimal Pocock & 0.1 & 1.33 (0.00) & 1.0000 & -0.91 (0.07) & 0.9996 \\
	Optimal Pocock & 0.2 & 1.40 (0.03) & 0.9999 & -1.03 (0.11) & 0.9989 \\
	\hline
	\end{tabular}
	\caption{Fitted exponents $\hat{p}$ in the model $\mathrm{EIF}_{K} \approx a + bK^{-p}$ under $H_{1}$ and $\hat{q}$ in the model $\mathrm{MIF}_{K} \approx a' + b'K^{q}$, with goodness-of-fit $R^{2}_{\mathrm{EIF}}$ and $R^{2}_{\mathrm{MIF}}$ for $K \in \{2, \ldots, 7\}$ at $(\alpha,\beta) \in \{(0.025,0.1),(0.025,0.2)\}$ under efficacy-only stopping. Designs are the jointly optimal HSD family design and the scheduling-only optimal OBF and Pocock designs of \citet{he2025} computed under the canonical Lan--DeMets formulation. Standard errors are obtained from nonlinear least squares. The underlying $\mathrm{EIF}_{K}$ and $\mathrm{MIF}_{K}$ values are shown in Figure~\ref{fig:331}.}
	\label{tab:331}
\end{table}

\begin{figure}[!ht]
	\centering
	\includegraphics[width=\linewidth]{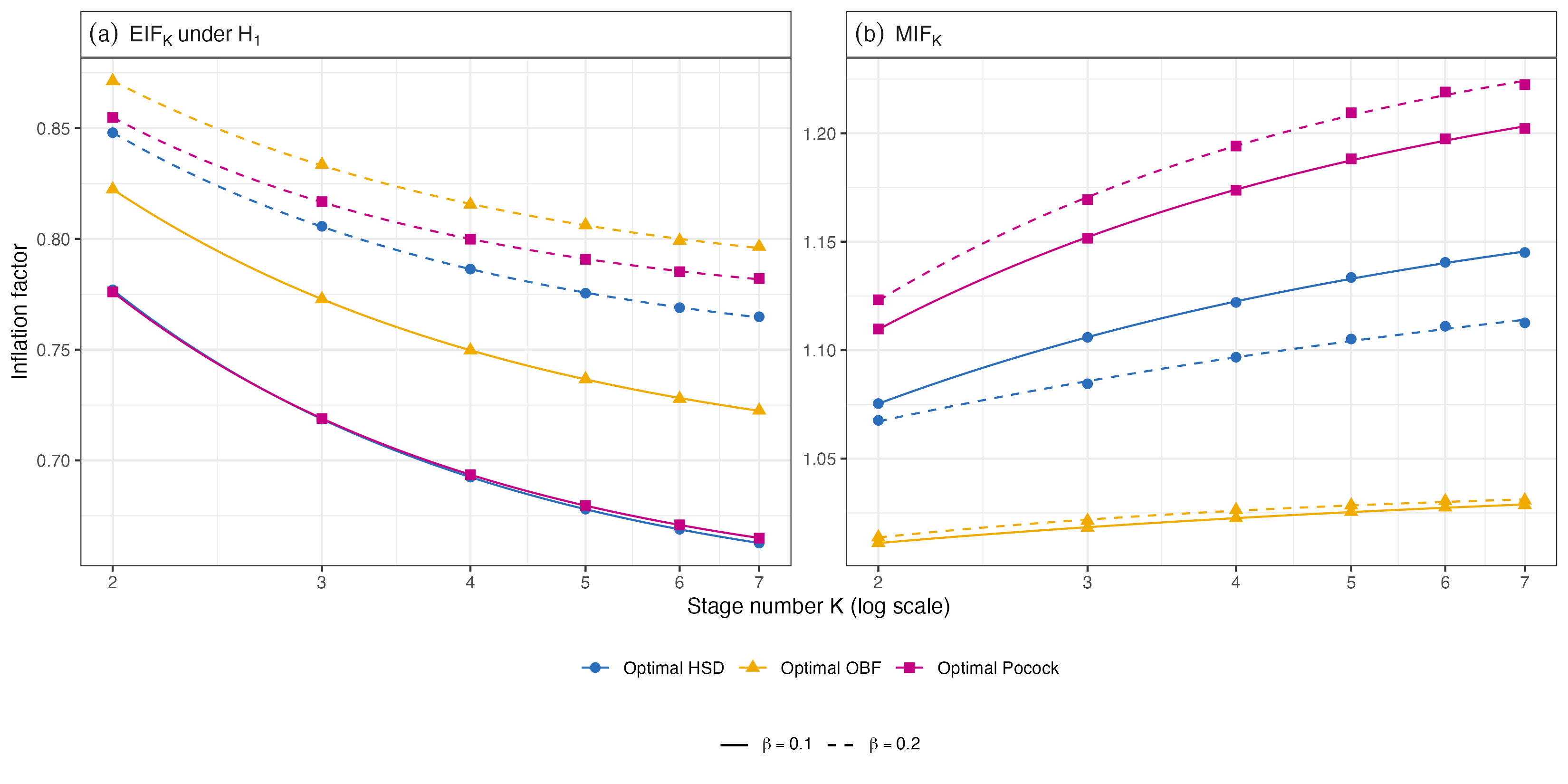}
	\caption{$\mathrm{EIF}_{K}$ under $H_{1}$ (panel~a) and $\mathrm{MIF}_{K}$ (panel~b) as functions of the stage number $K$ on a log $K$-scale for $(\alpha,\beta) \in \{(0.025,0.1),(0.025,0.2)\}$ under efficacy-only stopping. Three designs are overlaid: the jointly optimal HSD family, and the scheduling-only optimal OBF and Pocock designs of \citet{he2025} computed under the canonical Lan--DeMets formulation. Smooth curves show the fitted power laws $\mathrm{EIF}_{K} \approx a + b K^{-p}$ in panel~(a) and $\mathrm{MIF}_{K} \approx a' + b' K^{q}$ in panel~(b). The corresponding fitted exponents and standard errors are summarised in Table~\ref{tab:331}.}
	\label{fig:331}
\end{figure}

\begin{figure}[!ht]
	\centering
	\includegraphics[width=\linewidth]{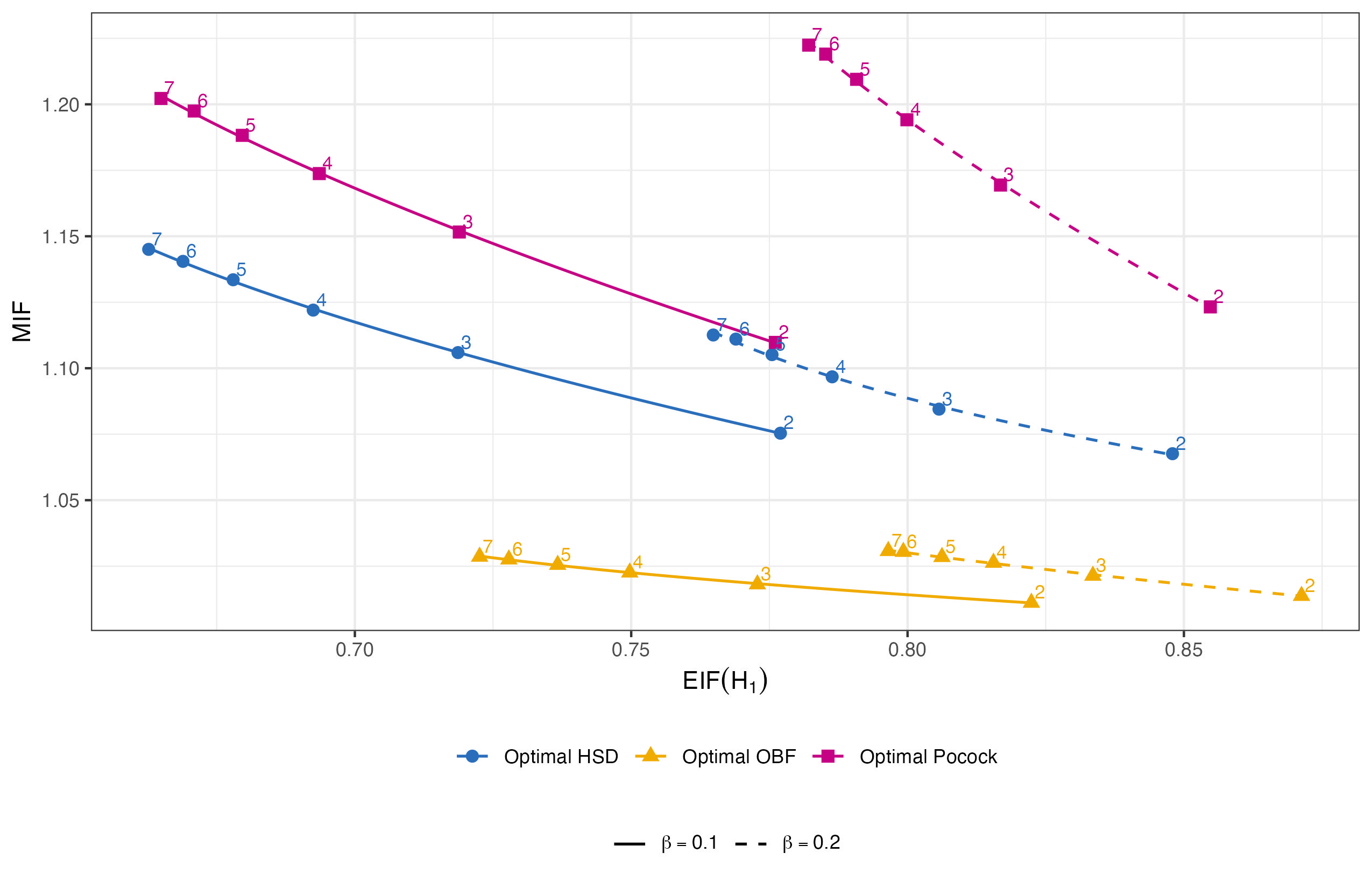}
	\caption{Trade-off frontier between the EIF under $H_{1}$ and the MIF as the stage number $K$ varies from 2 to 7 at $(\alpha,\beta) \in \{(0.025,0.1),(0.025,0.2)\}$ under efficacy-only stopping for the jointly optimal HSD family design and the scheduling-only optimal OBF and Pocock designs. Points represent the optimal designs at $K \in \{2,3,4,5,6, 7\}$ and are labelled by $K$. Curves are the smooth Pareto frontiers obtained by parametrising $K$ continuously in the fitted power laws $\mathrm{EIF}_{K} \approx a + bK^{-p}$ and $\mathrm{MIF}_{K} \approx a' + b'K^{q}$ from Table~\ref{tab:331}. The recommended $K^{\star}$ for the joint HSD family is $K^{\star}=6$ at both power levels.}
	\label{fig:332}
\end{figure}

The fitted EIF exponent for the jointly optimal HSD family design was $\hat{p} \approx 1.31$ at $\beta=0.1$ and $1.30$ at $\beta=0.2$. Under the canonical Lan--DeMets spending, the optimal OBF design yielded $\hat{p}$ in the range $1.20$--$1.28$ and the optimal Pocock design in the range $1.33$--$1.40$. The fitted exponent of the jointly optimal design therefore lay between those of the two scheduling-only designs, closer to that of the optimal Pocock design at $\beta=0.1$ and to that of the optimal OBF design at $\beta=0.2$. The EIF model fitted closely, with $R^{2}_{\mathrm{EIF}} \geq 0.9996$ for every design and power level.

The fitted MIF exponent $\hat{q}$ was negative for all three designs, so that under this parameterisation the MIF approaches its inflation ceiling $a'$ from below at rate $K^{-|q|}$. At $\beta=0.1$, $\hat{q}$ ranged from around $-0.61$ for optimal OBF, through $-0.80$ for the joint HSD family, to $-0.91$ for optimal Pocock. At $\beta=0.2$ it was close to $-1.0$ for both scheduling-only designs but only $-0.51$ for the joint HSD family, the least well determined of the six fits ($R^{2}_{\mathrm{MIF}}=0.9959$). By $K=6$, the MIF of every design lay within $0.06$ of its fitted ceiling $a'$. The marginal MIF increment from $K=6$ to $K=7$ never exceeded $0.005$. The MIF model fitted well for every design and power level, with $R^{2}_{\mathrm{MIF}} \geq 0.99$ throughout.

Figure~\ref{fig:332} shows the $(\mathrm{EIF}_{K}, \mathrm{MIF}_{K})$ trade-off frontier. The optimal Pocock frontier extended to lower EIF at the cost of substantially higher MIF (reaching $1.20$ at $K=7$ and $\beta=0.1$), whereas the optimal OBF frontier remained between $\mathrm{MIF}=1.01$ and $1.03$ but at higher EIF. The jointly optimal HSD frontier lay between the two in MIF but reached EIF values at or below those of optimal Pocock for $K \geq 3$.

The optimal $K$-stage GSD can be viewed as a discretisation of a continuous-time Brownian-motion optimal stopping problem of the kind studied by \citet{eales1992} and \citet{lai1973}. For a one-dimensional diffusion, restricting stopping to a grid of mesh $h$ changes the optimal value by order $h$ \citep{dupuis2005}, which corresponds to a $1/K$ rate of approach to the continuous-time optimum. The empirical exponents $\hat{p} \in [1.20, 1.40]$ were uniformly larger than this prediction. Each design approaches its own limit $a$, however, and only the unrestricted optimum is a discretisation of the continuous-time optimal stopping problem, so the ordering $\hat{p}_{\text{OBF}} < \hat{p}_{\text{Joint}} < \hat{p}_{\text{Pocock}}$ has no direct interpretation in that theory.

\subsection{The CORRECT trial}
\label{sec:34}
The redesign followed the configuration of Section~\ref{sec:261}, with the interim schedule and the efficacy stopping boundary jointly optimised under the HSD family.

\subsubsection{Joint optimum and event-count operating characteristics}
\label{sec:341}
Table~\ref{tab:341} reports the jointly optimal information rates at each nominal $K$ between 2 and 7, translated to event-count operating characteristics under the CORRECT OS configuration. By the invariance property of Section~\ref{sec:22}, these rates and the corresponding shape parameters carry over directly from Table~\ref{tab:311} at $\beta=0.1$, so no re-optimisation is needed for the survival setting. At $\beta=0.1$ the convex alpha-spending constraint was active at every $K$, so the optimal shape parameter equalled the boundary value $\hat{\gamma}=0$ throughout, corresponding to uniform spending. This is consistent with Section~\ref{sec:32}, where the constraint binds at $\beta=0.1$ but not at $\beta=0.2$.

Applying the selection rule of Section~\ref{sec:23} gave $K^{\star}=6$ for the jointly optimal HSD family design. With $d_{K}=652$, the default $\epsilon=0.05$ requires successive analyses to be at least $33$ events apart, and the smallest gap in the recommended schedule was $79$ events, so no analyses were merged. The margin on the tolerance criterion was narrow. With $K=7$ as the minimiser, at an expected number of events $\mathbb{E}(D \mid H_{1})=378.6$, the tolerance $\tau=0.01$ admits any design within $382.4$ events, and $K=6$ met this bound at $382.2$ events whereas $K=5$ did not ($387.4$ events). The slack at $K=6$ was therefore only $0.2$ events.

\begin{table}[!ht]
	\centering
	\begin{tabular}{cccccc}
	\hline
	$K$ & $\hat{t}_{1:K-1}$ (\%) & $d_{K}$ &
	$\mathbb{E}(D|H_{0})$ & $\mathbb{E}(D|H_{0}/H_{1})$ & $\mathbb{E}(D|H_{1})$ \\
	\hline
	2 & 51.0 & 614 & 610.6 & 569.0 & 443.9 \\
	3 & 37.7/65.7 & 632 & 626.6 & 570.9 & 410.6 \\
	4 & 31.5/51.4/72.7 & 641 & 635.2 & 572.6 & 395.7 \\
	5 & 27.8/44.5/61.4/77.6 & 648 & 641.3 & 574.0 & 387.4 \\
	\textbf{6} & 24.8/38.7/53.3/65.5/80.0 & 652 & 644.8 & 574.9 & 382.2 \\
	7 & 23.4/35.0/47.0/55.8/67.8/81.8 & 654 & 647.3 & 575.6 & 378.6 \\
	\hline
	\end{tabular}
	\caption{Event-count operating characteristics for the jointly optimal HSD family CORRECT redesign at each nominal $K$ between 2 and 7, at one-sided $\alpha=0.025$ and $\beta=0.1$ (90\% power). Reported are the jointly optimal information rates $\hat{\boldsymbol{t}}_{1:K-1}$ (\%), the maximum number of events $d_{K}$, and the expected number of events $\mathbb{E}(D \mid H_{h})$ under $H_{0}$, $H_{0}/H_{1}$ and $H_{1}$. The information rates and the shape parameter $\hat{\gamma}$ underlying these characteristics are taken from Table~\ref{tab:311} at $\beta=0.1$ by the invariance property of Section~\ref{sec:22}. The recommended $K^{\star}$ is highlighted in bold.}
	\label{tab:341}
\end{table}

Table~\ref{tab:342} reports the event-count operating characteristics for the candidate designs introduced in Section~\ref{sec:261} at $K=3$ (matching the original CORRECT configuration) and at $K^{\star}=6$, together with the original CORRECT design as a reference. At $K=3$, the jointly optimal HSD family design placed its interim analyses at $(37.7\%, 65.7\%)$, appreciably earlier than the schedule the trial actually followed at $(51.7\%, 74.2\%)$. Relative to the original design, the saving in expected events under $H_{1}$ was $68.7$ events ($410.6$ versus $479.3$, a $14.3\%$ reduction) at the cost of an increase in maximum events ($632$ versus $581$, an $8.8\%$ increment). At $K^{\star}=6$, the jointly optimal design further reduced $\mathbb{E}(D \mid H_{1})$ to $382.2$ events, saving a further $28.4$ events relative to its $K=3$ counterpart and $97.1$ events ($20.3\%$) relative to the original design, at $d_{K^{\star}}=652$ maximum events. Across the three optimal candidates at $K^{\star}=6$, the jointly optimal design saved $33.6$ events ($8.1\%$) under $H_{1}$ relative to the scheduling-only optimal OBF design ($415.8$ events). Against the scheduling-only optimal Pocock design, its expected number of events under $H_{1}$ was almost the same ($382.2$ versus $383.3$), with $32$ fewer maximum events ($652$ versus $684$). The gain in expected events over optimal Pocock was small because at 90\% power the convex alpha-spending constraint bound ($\hat{\gamma}=0$), which kept the joint design from spending alpha as early as the Pocock shape does. At 80\% power, where the constraint did not bind, the joint optimum was strictly better on both measures (Section~\ref{sec:31}). The Pareto dominance established on the information scale in Section~\ref{sec:31} therefore carries over to this configuration. Under $H_{0}$, each candidate design accumulated an expected number of events close to its maximum $d_{K}$, reflecting the small probability of early efficacy stopping when no treatment effect is present.

\begin{table}[!ht]
	\centering
	\begin{tabular}{lccccc}
	\hline
	Design & $K$ & $d_{K}$ & $\mathbb{E}(D|H_{0})$ & $\mathbb{E}(D|H_{0}/H_{1})$ & $\mathbb{E}(D|H_{1})$ \\
	\hline
	Original CORRECT & 3 & 581 & 579.4 & 555.0 & 479.3 \\
	Optimal OBF ($K=3$) & 3 & 582 & 580.0 & 550.1 & 441.5 \\
	Optimal Pocock ($K=3$) & 3 & 658 & 651.3 & 587.0 & 410.7 \\
	Optimal HSD joint ($K=3$) & 3 & 632 & 626.6 & 570.9 & 410.6 \\
	Optimal OBF ($K^{\star}=6$) & 6 & 587 & 584.8 & 545.7 & 415.8 \\
	Optimal Pocock ($K^{\star}=6$) & 6 & 684 & 675.7 & 595.9 & 383.3 \\
	Optimal HSD joint ($K^{\star}=6$) & 6 & 652 & 644.8 & 574.9 & 382.2 \\
	\hline
	\end{tabular}
	\caption{Event-count operating characteristics for the CORRECT redesign at $K=3$ (matching the original CORRECT two-interim configuration) and at the recommended $K^{\star}=6$: the maximum number of events $d_{K}$, and the expected number of events $\mathbb{E}(D \mid H_{h})$ under $H_{0}$, $H_{0}/H_{1}$ and $H_{1}$. The three candidate redesigns are the jointly optimal HSD family design, the scheduling-only optimal OBF design and the scheduling-only optimal Pocock design (each of the latter two using the optimal information rates of \citet{he2025} with the canonical Lan--DeMets alpha-spending function of the corresponding shape). The original CORRECT design is included as a reference at its realised information rates $(51.7\%, 74.2\%)$, evaluated on its efficacy schedule alone so that the comparison with the efficacy-only redesigns is like for like. Its only efficacy boundary is at $74.2\%$, as the analysis at $51.7\%$ was for futility. Recomputed for this efficacy-only schedule, its maximum is $581$ events rather than the planned $582$.}
	\label{tab:342}
\end{table}

\subsubsection{Calendar-time operating characteristics at the CORRECT baseline}
\label{sec:342}
Tables~\ref{tab:343} and~\ref{tab:344} summarise the calendar-time operating characteristics at the CORRECT baseline described in Section~\ref{sec:262}. Each design is sized separately to its own maximum number of events $d_{K}$, so the MSS varies across designs in proportion to $d_{K}$. Under $H_{1}$, the expected stopping time of every redesign candidate preceded the end of its enrolment, although its final analysis under $H_{1}$ would fall five to six and a half months after enrolment ends. The ESS under $H_{1}$ therefore fell about 10\% to 30\% below the MSS for every redesign candidate, and the efficiency gains apply to the number of patients as well as to trial duration. For the jointly optimal design at $K^{\star}=6$, $\mathbb{E}(N \mid H_{1})=573$ against $N_{K}=774$. The original design, by contrast, had an expected stopping time of $26.3$ months under $H_{1}$, after its enrolment ended at $25$ months, so its ESS was only $3\%$ below its MSS.

\begin{table}[!ht]
	\centering
	\begin{tabular}{lccccc}
	\hline
	Design & $K$ & $N_{K}$ & $\mathbb{E}(N|H_{0})$ & $\mathbb{E}(N|H_{0}/H_{1})$ & $\mathbb{E}(N|H_{1})$ \\
	\hline
	Original CORRECT & 3 & 690 & 689 & 681 & 668 \\
	Optimal OBF ($K=3$) & 3 & 691 & 690 & 676 & 625 \\
	Optimal Pocock ($K=3$) & 3 & 782 & 776 & 726 & 595 \\
	Optimal HSD joint ($K=3$) & 3 & 751 & 746 & 706 & 593 \\
	Optimal OBF ($K^{\star}=6$) & 6 & 697 & 696 & 678 & 607 \\
	Optimal Pocock ($K^{\star}=6$) & 6 & 813 & 805 & 742 & 575 \\
	Optimal HSD joint ($K^{\star}=6$) & 6 & 774 & 768 & 717 & 573 \\
	\hline
	\end{tabular}
	\caption{Sample size operating characteristics for the CORRECT redesign at the CORRECT baseline, with each design sized separately so that the MSS $N_{K}$ scales with its maximum number of events $d_{K}$ (size-for-events sizing based on the original CORRECT design, with $N_{K}=690$ over a 25-month accrual period). Reported are the MSS $N_{K}$ and the ESS $\mathbb{E}(N \mid H_{h})$ under $H_{0}$, $H_{0}/H_{1}$ and $H_{1}$. The compared designs and the modelling assumptions (exponential survival with control-arm median OS of $4.5$ months, planned accrual of 30 patients per month after a four-month ramp-up, 3\% annual drop-out, event-driven analyses) are detailed in Sections~\ref{sec:261} and~\ref{sec:262}.}
	\label{tab:343}
\end{table}

\begin{table}[!ht]
	\centering
	\begin{tabular}{lccccc}
	\hline
	Design & $K$ & $T_{K}$ (mo) & $\mathbb{E}(T|H_{0})$ & $\mathbb{E}(T|H_{0}/H_{1})$ & $\mathbb{E}(T|H_{1})$ \\
	\hline
	Original CORRECT & 3 & 31.3 & 28.9 & 28.9 & 26.3 \\
	Optimal OBF ($K=3$) & 3 & 31.4 & 28.9 & 28.7 & 24.9 \\
	Optimal Pocock ($K=3$) & 3 & 33.6 & 31.1 & 29.6 & 23.5 \\
	Optimal HSD joint ($K=3$) & 3 & 32.8 & 30.3 & 29.1 & 23.5 \\
	Optimal OBF ($K^{\star}=6$) & 6 & 31.5 & 29.1 & 28.4 & 23.8 \\
	Optimal Pocock ($K^{\star}=6$) & 6 & 34.3 & 31.8 & 29.7 & 22.2 \\
	Optimal HSD joint ($K^{\star}=6$) & 6 & 33.4 & 30.9 & 29.1 & 22.2 \\
	\hline
	\end{tabular}
	\caption{Trial duration operating characteristics for the CORRECT redesign at the CORRECT baseline (companion to Table~\ref{tab:343}, with compared designs, modelling assumptions and the size-for-events sizing convention given in Section~\ref{sec:262}). Reported are the MTD $T_{K}$ and the ETD $\mathbb{E}(T \mid H_{h})$ under $H_{0}$, $H_{0}/H_{1}$ and $H_{1}$, all in months.}
	\label{tab:344}
\end{table}

At $K=3$, the jointly optimal design reduced the ESS under $H_{1}$ from $668$ patients for the original design to $593$, and the ETD from $26.3$ to $23.5$ months, at an MSS of $751$ rather than $690$. At the recommended $K^{\star}=6$ it reduced the ESS further to $573$ patients and the ETD to $22.2$ months, saving $95$ patients ($14.2\%$) and $4.1$ months ($15.6\%$) relative to the original design at an MSS of $774$. Against the scheduling-only optimal OBF design at $K^{\star}=6$, the joint optimum saved $34$ patients on the ESS ($573$ versus $607$) at the cost of $77$ additional patients on the MSS. Against the scheduling-only optimal Pocock design it achieved the same ETD ($22.2$ months) and a nearly identical ESS ($573$ versus $575$) while committing $39$ fewer patients on the MSS ($774$ versus $813$).

Evaluated with its non-binding futility rule in place, as described in Section~\ref{sec:261}, at its planned $582$ events and $N_{K}=690$, the original CORRECT design had $\mathbb{E}(N \mid H_{0})=566$ and $\mathbb{E}(N \mid H_{1})=666$. The reduction under $H_{0}$ reflects early stopping for futility, which the efficacy-only redesigns cannot do (Section~\ref{sec:261}). Section~\href{run:./ZH2023_Supplemental_Material.pdf}{S7} reports a redesign of the CLEOPATRA trial, in which median survival is long relative to the accrual period rather than short as in CORRECT.

\section{Discussion}
\label{sec:4}
We have jointly optimised the scheduling of interim analyses and the shape of the stopping boundaries in group sequential trials, minimising the ESS under a specified alternative hypothesis while controlling overall type I and type II error rates. 
We report three main findings. 
First, the joint optima under the HSD family \citep{hwang1990}, with the shape parameter unconstrained, agreed with the unrestricted numerical optima reported for the two- and three-stage settings in \citet{brittain1993}. Therefore, for two- and three-stage efficacy-only designs, restricting the boundary shape to a one-parameter spending function family causes a negligible efficiency loss relative to free optimisation over critical values. Second, for $K \geq 3$ the jointly optimal design Pareto-dominated the scheduling-only optimal Pocock design of \citet{he2025}. At 90\% power, it achieved a nearly identical ESS saving under $H_{1}$ while incurring strictly smaller MSS inflation. At 80\% power, it achieved both a smaller EIF under $H_{1}$ and a smaller MIF.
Against the scheduling-only optimal OBF design, the jointly optimal design had a smaller EIF under $H_{1}$ at the cost of a larger MIF. Third, at 90\% power, the joint optimum lay on the boundary of the convex alpha-spending feasible region for all stage numbers from two to seven, corresponding to uniform alpha-spending.
At 80\% power the optimum moved interior and, under the refinement requirements, grew more conservative as the number of interim looks increased.
Across both power levels, the HSD and KD families yielded numerically indistinguishable operating characteristics. 

The convex alpha-spending constraint excludes Pocock-like aggressive early spending while retaining shapes from uniform to strongly conservative, so it restricts the joint optimisation to clinically sensible designs without sacrificing meaningful flexibility in the boundary shape. The unconstrained optima of Section~\ref{sec:32} confirm that the constraint is binding: at 90\% power the unconstrained HSD optimum lies at $\hat{\gamma}\approx 0.76$ ($K=2$) and $\hat{\gamma}\approx 0.56$ ($K=3$), in the Pocock-like half of the parameter space. The same qualitative conclusion was reached by \citet{ham2026}, who optimised the rejection criterion directly and likewise found that the unconstrained optimum spends the alpha budget aggressively at early analyses. Section~\ref{sec:32} quantifies its cost at $K \in \{2,3\}$ as at most $0.0012$ in EIF under $H_{1}$.

The stage-merging rule addresses a practical difficulty that can arise whenever the interim schedule is optimised for a prespecified number of analyses. The scheduling-only optimal OBF design of \citet{he2025} at $K=7$ and $\beta=0.2$, for example, has two analyses only $4.2\%$ of the information apart. When the prespecified number of analyses is larger than the problem effectively requires, the optimiser may drive adjacent information rates so close together that the corresponding interim analyses become operationally indistinguishable. Suppressing such near-coincident analyses through a penalty would distort the optimum, so we use them to determine the appropriate effective stage number. At the default threshold $\epsilon=0.05$, the effective stage number equalled the nominal stage number in every efficacy-only and efficacy plus non-binding futility configuration with $K \le 7$. Within that range, merging occurred only at the larger threshold $\epsilon=0.10$, and only at $K=7$. The eight-stage efficacy-only optima of Section~\ref{sec:31} did merge at the default threshold at 80\% power under both families, so, for the joint designs, schedules with near-coincident analyses begin to appear once $K$ exceeds $7$. The rule is therefore needed in practice, and nothing in the formulation of Eq.~(\ref{eqn:2203}) prevents such schedules at larger $K$ or under other spending function families.

As the stage-merging rule did not bind at the default threshold for $K \le 7$, the recommendation of $K^{\star}=6$ at both 80\% and 90\% power rests on the tolerance criterion alone. The ESS at $K=6$ lay within a factor $(1+\tau)$ of the minimum over $K \le 7$, with $\tau=0.01$, whereas the ESS at $K=5$ did not. At 90\% power, extending the range to $K=8$ raised $K^{\star}$ to $7$. The margins are narrow enough that $\tau$ should be reported and varied explicitly rather than treated as a fixed default. 
By the tolerance criterion, the efficiency gain was concentrated in the first five interim analyses, that is in six analyses in total. \citet{todd2001} suggested that planning for between four and eight interim analyses is sensible on administrative grounds, and our efficiency results place the choice at the lower end of that range. This complements the observation in \citet{he2025} that marginal ESS reductions diminish substantially after three interim analyses. A rigorous asymptotic characterisation of the scaling exponents of Section~\ref{sec:33} would require extending the analysis of \citet{dupuis2005} to the constrained-optimisation setting considered here, which we leave for future work.

The operating characteristics presented in the main text focus on efficacy-only GSDs. The corresponding analysis for GSDs with both efficacy and non-binding futility boundaries introduces a second shape parameter into the optimisation, denoted by $\gamma^{*}$ under the HSD family and $\rho^{*}$ under the KD family. These results are reported in full in Section~\href{run:./ZH2023_Supplemental_Material.pdf}{S5}. Much of the structure established in Section~\ref{sec:3} carries over. The invariance property continues to hold, the two spending function families agreed on the EIF under $H_{1}$ to within $0.001$, and the EIF again followed a power law in $K$. The families differed by up to $0.023$ in MIF at $80\%$ power, however, and the fitted EIF scaling exponents were milder than for efficacy-only designs. Adding a non-binding futility boundary gave a modest further reduction in the ESS under $H_{1}$, larger at $80\%$ power than at $90\%$, at the cost of an increase in MSS that grows with the type II error rate. It also shifted the recommended $K^{\star}$ upwards by one stage in three of the four family and power combinations considered, one of them by a margin below the numerical error of the EIF. $K^{\star}=7$ was then the largest $K$ searched and so a lower bound. Under the cross-$K$ requirements of the refinement step (Section~\href{run:./ZH2023_Supplemental_Material.pdf}{S2}), the two shape parameters move in opposite directions as $K$ increases. Because the EIF under $H_{1}$ is nearly flat in the futility shape, this direction is set by the refinement rather than by the ESS. The efficacy shape lies at or near the uniform-spending edge of the feasible region and becomes slightly more conservative, whereas the futility shape is more conservative than uniform spending and becomes less so. The conservative level of early futility spending, though not its trend in $K$, is consistent with the asymmetry between the two kinds of early stop. An early efficacy stop terminates the trial successfully, whereas an early futility stop forfeits the chance to detect a true treatment effect. The selected designs keep efficacy spending at or near uniform and spend type II error conservatively at early looks.

As the joint optimum is invariant to the endpoint type and effect size, it can be tabulated once and reused. Table~\ref{tab:311} covers efficacy-only designs, and the corresponding table in Section~\href{run:./ZH2023_Supplemental_Material.pdf}{S5} covers designs with both efficacy and non-binding futility stopping boundaries. They are reference optimal designs under the framework stated here, namely the canonical joint distribution, the HSD or KD family, the convex error-spending constraint, minimisation of the ESS under $H_{1}$, and the tolerance $\tau$ that defines $K^{\star}$. They are not universally optimal group sequential schedules. 

The practical utility of the proposed approach has been shown through a redesign of the CORRECT trial in previously treated metastatic colorectal cancer, involving survival endpoints under proportional hazards. We compared the jointly optimal design with the scheduling-only optimal OBF and Pocock designs of \citet{he2025} on both the event-count and calendar-time scales, with each design sized separately to its own maximum number of events under the standard size-for-events convention. At $K^{\star}=6$, the jointly optimal HSD family design saved about $97$ expected events under $H_{1}$ relative to the original CORRECT design and $34$ relative to optimal OBF, and matched optimal Pocock on expected events with $32$ fewer maximum events. On the calendar-time scale, its Pareto dominance over optimal Pocock carried over to the patient scale, with an ESS under $H_{1}$ of $573$ against $575$ and an MSS of $774$ against $813$, at the same ETD under $H_{1}$ ($22.2$ months). Because the median survival is short relative to the accrual period, the expected stopping time under $H_{1}$ preceded the end of enrolment, so these savings were realised in patients as well as in trial duration.

The applicability to survival endpoints under proportional hazards, set out in Section~\ref{sec:21} and illustrated through the CORRECT redesign in Sections~\ref{sec:26} and \ref{sec:34}, relies on the asymptotic normality of the log-rank test statistic and on the information-time relation $I_{k}=d_{k}r/(1+r)^{2}$. On the events scale, Proposition~\ref{prop:2201} continues to apply, so the reference tables of Section~\ref{sec:31} carry over directly, at any allocation ratio, with ``events'' substituted for ``sample size''. For operational planning, the calendar-time operating characteristics depend additionally on the accrual pattern, the baseline hazard and the hazard ratio, as Section~\ref{sec:342} makes explicit. Where the optimised endpoint is not the primary one, its interim schedule is commonly inherited from the primary endpoint's, and relocating its looks would require either decoupling the two analyses or scheduling both endpoints jointly. Extending the joint optimisation to a hierarchical family of endpoints is therefore a direction for future work. A further extension is to relax the proportional hazards assumption, for example through piecewise-exponential hazards with time-varying treatment effects, or through weighted log-rank tests that preserve approximate normality under non-proportional hazards. These retain an approximately multivariate normal joint distribution with the same covariance structure, but the drift no longer grows in proportion to $\sqrt{t_{k}}$, so the optimisation machinery carries over with a modified drift and the reference tables would need to be recomputed.

Like that of \citet{he2025}, the formulation adopted here assumes that treatment outcomes are observed immediately while enrolment is halted at the interim analysis. This is uncommon in trial practice, and relaxing the assumption is a related but more immediate generalisation. Follow-up duration and continuous enrolment for non-survival endpoints can be accommodated by replacing the information rates in the ESS with their pipeline-adjusted counterparts, as set out by \citet{togo2013} and \citet{he2025}. The joint optimisation of Eq.~(\ref{eqn:2203}) would extend to that objective without structural change, and the efficiency gains reported here would be attenuated to the extent that patients already in the pipeline cannot be recovered by an early stop.

We minimise the ESS under $H_{1}$ as the primary optimality criterion because it captures the efficiency gain of a GSD when the trial stops early for efficacy. The true effect size may in practice be smaller than that assumed under $H_{1}$. Several alternatives can replace the objective in Eq.~(\ref{eqn:2203}) without structural change. These include a weighted combination of ESS values across selected effect sizes \citep{xi2017}, an integration of the ESS over a distribution of plausible effect sizes \citep{lokhnygina2008}, and a weighted combination of the ESS, the MSS and the standard deviation of the sample size \citep{grayling2021}. No explicit constraint on the MSS is imposed, but the MSS is reported alongside the ESS throughout, so the inflation cost of each design is visible.

The stopping boundaries are characterised through the HSD family of \citet{hwang1990} and the KD family of \citet{kim1987}. Other parametric families \citep{wang1987,emerson1989,pampallona1994,jennison1987} can be substituted in the same way. A more ambitious extension is to optimise the interim schedule jointly with the full set of critical values $\boldsymbol{u}_{1:K}$ and $\boldsymbol{l}_{1:K}$, without restricting to a parametric family. This has been done for tests with binding futility boundaries and up to six analyses \citep{eales1992,jennison2006}, and for two- and three-stage efficacy-only designs \citep{brittain1993}. \citet{ham2026} optimised the rejection criterion directly at fixed analysis times through a sample-average approximation solved by mixed-integer linear programming. Combining that formulation with optimisation over the interim schedule would provide a way to obtain the fully unconstrained joint optimum at larger $K$, at substantially greater computational cost. The benchmark of Section~\ref{sec:32} bounds the gain from such an extension only at $K \in \{2,3\}$, where the parametric restriction costs at most $0.001$ in EIF. For tests with binding futility boundaries and up to six analyses, \citet{jennison2006} found this gain to be at most 2\% of the fixed sample size, even relative to a monomial spending family with only the first group size free.

The approach provides a practical basis for planning interim looks across a wide range of group sequential trials, jointly determining the interim analysis timings and the stopping boundary shape together with the effective number of analyses through the tolerance-based selection criterion. It is implemented in the extended \texttt{OptimInterim}, and it can achieve efficiency gains beyond those reported in \citet{he2025}.



\subsection*{Data availability statement}
This work uses no individual participant data. The design parameters underlying the case studies are taken from the published reports of the CORRECT trial \citep{grothey2013} and its regulatory statistical review \citep{fda2012}, and of the CLEOPATRA trial (Section~\href{run:./ZH2023_Supplemental_Material.pdf}{S7}). \texttt{OptimInterim} is openly available at \url{https://github.com/zhangyi-he/GSD_OptimInterim}, together with the stored designs reported in this work, the optimisation and refinement code used to produce them, and the code that generates every table and figure from them.

\subsection*{Financial disclosure}
The authors declare no funding associated with the work presented in this article.

\subsection*{Conflict of interest}
The authors declare no potential conflict of interests.






\bibliography{ZH2023_Manuscript}

\end{document}